%% file: main.tex
\documentclass[11pt]{article}
\usepackage[utf8]{inputenc}

\usepackage{amsmath}
\usepackage{amsthm}
\usepackage{amssymb}
\usepackage{algorithm}
\usepackage{color}
\usepackage{xcolor}
\usepackage[english]{babel}
\usepackage{bbm}

\usepackage{graphicx}
\usepackage{caption}
\usepackage{subcaption}

\usepackage{enumitem}

\usepackage[square,sort,comma,numbers]{natbib}
\usepackage{wrapfig,epsfig}
\usepackage{psfrag}
\usepackage{epstopdf}
\usepackage{url}
\usepackage{graphicx}
\usepackage{color}
\usepackage{epstopdf}
\usepackage{algpseudocode}
\usepackage[hidelinks,pdfencoding=auto,psdextra]{hyperref}
\hypersetup{colorlinks=false}
\hypersetup{
	colorlinks,
	linkcolor={red!40!gray},
	citecolor={blue!40!gray},
	urlcolor={blue!70!gray}
}
\usepackage[margin=1in]{geometry}
\usepackage{float}
\usepackage{multirow}

\newtheorem{theorem}{Theorem}[section]
\newtheorem{lemma}[theorem]{Lemma}

\newtheorem{definition}[theorem]{Definition}

\newtheorem{proposition}[theorem]{Proposition}

\newtheorem{fact}[theorem]{Fact}
\newtheorem{remark}[theorem]{Remark}
\newtheorem{claim}[theorem]{Claim}

\newcommand{\eps}{\varepsilon}

\newcommand{\supp}{\mathsf{supp}}

\DeclareMathOperator*{\E}{\mathbb{E}}

\newcommand{\authnote}[3]{\textcolor{#2}{{\sf (#1's Note: {\sl{#3}})}}}
\newcommand{\xnote}{\authnote{Xin}{magenta}}

\newenvironment{proofof}[1]{\bigskip \noindent {\it Proof of #1.}\quad }
{\qed\par\vskip 4mm\par}

\begin{document}

\title{Improved Decoding of Expander Codes}
\author{Xue Chen\thanks{\tt{xuechen1989@ustc.edu.cn}, University of Science and Technology of China \& The CAS Key Laboratory of Wireless-Optical Communications, USTC. Part of this work is done while the author was at George Mason University.}
\and Kuan Cheng\thanks{
\tt{ckkcdh@pku.edu.cn}, Peking University.
}
\and Xin Li\thanks{
\tt{lixints@cs.jhu.edu}, Johns Hopkins University. Supported by NSF CAREER Award CCF-1845349 and NSF Award CCF-2127575.
}
\and Minghui Ouyang\thanks{\tt{ouyangminghui1998@gmail.com}, Peking University.
}
}

\date{}

\maketitle

\begin{abstract}
We study the classical expander codes, introduced by Sipser and Spielman \cite{SS96}. Given any constants $0< \alpha, \eps < 1/2$, and an arbitrary bipartite graph with $N$ vertices on the left, $M < N$ vertices on the right, and left degree $D$ such that any left subset $S$ of size at most $\alpha N$ has at least $(1-\eps)|S|D$ neighbors, we show that the corresponding linear code given by parity checks on the right has distance at least roughly $\frac{\alpha N}{2 \eps}$. This is strictly better than the best known previous result of $2(1-\eps) \alpha N$ \cite{Sudan2000note, Viderman13b} whenever $\eps < 1/2$, and improves the previous result significantly when $\eps$ is small. Furthermore, we show that this distance is tight in general, thus providing a complete characterization of the distance of general expander codes.

Next, we provide several efficient decoding algorithms, which vastly improve previous results in terms of the fraction of errors corrected, whenever $\eps < \frac{1}{4}$. Finally, we also give a bound on the list-decoding radius of general expander codes, which beats the classical Johnson bound in certain situations (e.g., when the graph is almost regular and the code has a high rate). 

Our techniques exploit novel combinatorial properties of bipartite expander graphs. In particular, we establish a new size-expansion tradeoff, which may be of independent interests.
\end{abstract}


\input{intro}

\input{preli}

\input{distance}

\input{find}


\input{guessingwithflips}

\input{Improved_unique_decoding}

\input{list_decoding_bound}

\input{open}
\bibliographystyle{alpha} 
\bibliography{test}

\appendix



\input{appendix}

\end{document}

%% file: intro.tex

\section{Introduction}
Expander codes \cite{SS96} are error-correcting codes derived from bipartite expander graphs that are notable for their ultra-efficient decoding algorithms. In particular, all known asymptotically good error-correcting codes which admit (almost) linear-time decoding algorithms for a constant fraction of adversarial errors are based on expander codes. At the same time, expander codes are closely related to low-density parity-check (LDPC) codes \cite{Gallager1963} --- a random LDPC code is an expander code with high probability. Over the last twenty years, LDPC codes have received increased attention (\cite{FWK05,FMSSW07,ADS12,dimakis2012ldpc,JNNSM20LDPC} to name a few) because of their practical performance. Along this line of research, the study of decoding algorithms for expander codes, such as belief-propagation \cite{Gallager1963,SS96, luby1998improved}, message-passing \cite{richardson2001the}, and linear programming \cite{FWK05, FMSSW07, Viderman13}, has laid theoretical foundations and sparked new lines of inquiry for LDPC codes.

In this work, we consider expander codes for adversarial errors. Briefly, given a bipartite graph $G$ with $N$ vertices of degree $D$ on the left, we say it is an $(\alpha N, (1-\eps)D)$ expander if and only if any left subset $S$ with size at most $\alpha N$ has at least $(1-\eps)D \cdot |S|$ distinct neighbors. The code $\mathcal{C}$ of an expander $G$ assigns a bit to each vertex on the left and views each vertex on the right as a parity check over its neighbors. A codeword $C \in \mathcal{C}$ is a vector in $\{0,1\}^N$ that satisfies all parity checks on the right. Moreover, the distance of $\mathcal{C}$ is defined as the minimum Hamming distance between all pairs of codewords. We defer the formal definitions of expanders and expander codes to Section~\ref{sec:preli}. For typical applications, the parameters $\alpha,\eps$ and $D$ are assumed to be constants, and there exist explicit constructions (e.g., \cite{CRVW}) of such expander graphs with $M<N$.

For expander codes defined by $(\alpha N,(1-\eps)D)$-expanders, the seminal work of Sipser and Spielman \cite{SS96} gave the first efficient algorithm to correct a constant fraction (i.e., $(1-2\eps) \cdot \alpha N$) of errors, when $\eps < 1/4$. In fact, their algorithms are super efficient --- they provide a linear time algorithm called belief-propagation and a logarithmic time parallel algorithm with a linear number of processors. 
Subsequently, Feldman et al.~\cite{FMSSW07} and Viderman \cite{Viderman13, Viderman13b} provided improved algorithms to correct roughly $\frac{1-3\eps}{1-2\eps}\cdot \alpha N$ errors, when $\eps < 1/3$. This fraction of error is strictly larger than that of \cite{SS96} whenever $\eps < 1/4$. Viderman \cite{Viderman13b} also showed how to correct $N^{\Omega_{D,\eps,\alpha}(1)}$ errors when $\eps \in [1/3,1/2)$, and that $\eps < 1/2$ is necessary for correcting even $1$ error.\ However, the following basic question about expander codes remains unclear.

\paragraph{Question:}\emph{What is the best distance bound one can get from an expander code defined by arbitrary $(\alpha N,(1-\eps)D)$-expanders?}\\

This question is important since it is well known that for unique decoding, the code can and can only correct up to half the distance number of errors. In \cite{SS96}, Sipser and Spielman showed that the distance of such expander codes is at least $\alpha N$, while a simple generalization improves this bound to $2(1-\eps)\alpha N$ (see e.g., \cite{Sudan2000note} and \cite{Viderman13b}). Perhaps somewhat surprisingly, this simple bound is the best known distance bound for an arbitrary expander code. In fact, Viderman \cite{Viderman13b} asserted that this is the best distance bound one can achieve based only on the expansion property of the graph, and hence when $\eps$ converges to $0$, the number of errors corrected in \cite{Viderman13b}, $\frac{1-3\eps}{1-2\eps}\cdot \alpha N$ converges to the half distance bound. Yet, no evidence was known to support this claim. Thus it is natural to ask whether any improvement is possible, and if so, can one design efficient algorithms to correct more errors?

\subsection{Our Results}
\paragraph{Distance of expander codes.} In this work, we give affirmative answers to the above questions. Our first result shows that the best distance bound of expander codes defined by arbitrary $(\alpha N, (1-\eps)D)$-expanders is roughly $\frac{\alpha N}{2 \eps}$.

\begin{theorem}\label{thm:infor_distance}[Informal versions of Theorem~\ref{thm:dist_expander} and Theorem~\ref{thm:tight_distance}]
Given any $(\alpha N, (1-\eps)D)$-expander, let $\mathcal{C}$ be the expander code defined by it. The distance of $\mathcal{C}$ is at least $\frac{\alpha}{2\eps} \cdot N - O_{\eps}(1)$.

Moreover, for any constant $\eta>0$ there exists an $(\alpha N, (1-\eps)D)$-expander whose expander code has distance at most $(\frac{\alpha}{2\eps}+\eta) \cdot N$.
\end{theorem}

We remark that the bound $\frac{\alpha}{2\eps} \cdot N$ is always larger than the previous bound $2(1-\eps)\alpha N$ since we always have $\eps< 1/2$ in expander codes. For small $\eps$, this improves upon the previous bound by a factor of $1/4\eps$ roughly, which can be quite significant. 

\paragraph{Decoding algorithms.} Next we consider algorithms to correct more errors. Given the above bound on the distance of expander codes, the natural goal is to design efficient algorithms that can correct $\Theta(\alpha/\eps) \cdot N$ errors. We achieve this goal for all $\eps<1/4$. 

\begin{theorem}\label{thm:inf_decode_dist}[Informal version of Theorem~\ref{thm:MultipleGuessWithFlipsFinal}]
Given any constants $\alpha, \eta>0$ and $0< \eps<1/4$, there exist a linear time algorithm that for any expander code defined by an $(\alpha N, (1-\eps)D)$-expander, correct up to $(\frac{3 \alpha}{16 \eps}-\eta) \cdot N$ adversarial errors.
\end{theorem}
The bound $\frac{3 \alpha}{16 \eps} \cdot N$ is larger than all previous bounds for $\eps<1/4$ by at least a constant factor. For example, when $\eps$ is close to $1/4$, all previous works \cite{SS96,FMSSW07,Viderman13b} can only correct roughly $\frac{\alpha}{2} \cdot N$ errors, while our algorithm can correct roughly $\frac{3}{4} \cdot \alpha N$ errors. When $\eps$ is smaller, the improvment is even more significant, as no previous work can correct more than $\alpha N$ errors. 
On the other hand, given Theorem~\ref{thm:infor_distance}, one can hope for correcting roughly $\frac{\alpha}{4\eps} \cdot N$ errors, so  Theorem~\ref{thm:inf_decode_dist} falls slightly short of achieving it. 

Actually, we can correct more errors when $\eps$ is small. For example, when $\eps<0.08$, our algorithm in~Section~\ref{sec:improve_decoding} can correct more than $\frac{0.2\alpha}{\eps} \cdot N$ errors. We summarize all our results informally in Table~\ref{table:distance}, compared to the previous best results of \cite{FMSSW07,Viderman13b}.
\begin{table}[h!]
\centering
\begin{tabular}[b]{|c | c | c |c|} 
 \hline 
  & $\eps \in (0,\frac{3-2\sqrt{2}}{2})$ & $\eps \in [\frac{3-2\sqrt{2}}{2},1/8)$ & $\eps \in [1/8,1/4)$ \\ 
 [1ex] \hline 
 Distance from Theorem~\ref{thm:infor_distance} of this work & $\frac{1}{2 \eps} \cdot \alpha N$ & $\frac{1}{2 \eps} \cdot \alpha N$ & $\frac{1}{2 \eps} \cdot \alpha N$ \\ 
 [1ex] \hline 
 Decoding radius from \cite{FMSSW07,Viderman13b}  & $\frac{1-3\eps}{1- 2 \eps} \cdot \alpha N$ & $\frac{1-3\eps}{1- 2 \eps} \cdot \alpha N$ & $\frac{1-3\eps}{1- 2 \eps} \cdot \alpha N$ \\
 [1ex] \hline
 \multirow{2}{16em}{Decoding radius from this work} & $\frac{\sqrt{2}-1}{2\eps} \cdot \alpha N$ & $\frac{1-2\eps}{4\eps} \cdot \alpha N$ & $\frac{3}{16\eps} \cdot \alpha N$ \\
 & from Theorem~\ref{thm:guessexpansion} & from Theorem~\ref{thm:guessexpansion} & from Theorem~\ref{thm:MultipleGuessWithFlipsFinal}\\
 \hline
\end{tabular}
\caption{Summary of the distance and decoding radii for $\eps$.}
\label{table:distance}
\end{table}


\paragraph{List decoding.} Finally, we consider the list-decodability of expander codes. List decoding, introduced by Elias \cite{Elias57listdecoding} and Wozencraft \cite{Wozencraft58} separately, is a relaxation of the classical notion of unique decoding. In this setting, the decoder is allowed to output a small list of candidate codewords that include all codewords within Hamming distance $\rho N$ of the received word. Thus, the list decoding radius $\rho N$ could be significantly larger than half of the distance. For example, a very recent work by Mosheiff et al.~\cite{JNNSM20LDPC} shows random LDPC codes have list decoding radii close to their distance. In this setting, the classical Johnson bound shows that \emph{any} binary code with distance $d$ is list-decodable up to radius $r=\frac{N-\sqrt{N(N-2d)}}{2}$ with list size $N^{O(1)}$. If we set the Johnson bound $r$ as the baseline, a natural question is whether expander codes can list-decode more than $r$ errors given the distance $d=\frac{\alpha}{2\eps}\cdot N$?

In Section~\ref{sec:list_decode}, we consider expander codes defined by expanders that has a maximum degree $D_{\max}=O(1)$ on the right, like LDPC codes. Our main results provide an alternative bound on the list-decoding radius of such codes, and show that it is strictly better than the Johnson bound when $\alpha/\eps$ is small and the right hand side is also almost regular, i.e., $D_{\max} \approx D_R$, where $D_R$ is the \emph{average} right degree. 

\begin{theorem}\label{thm:infor_list_decode}[Informal version of Theorem~\ref{thm:list_decoding_radius}]
Given any $(\alpha N,(1-\eps)D)$-expander with regular degree $D$ on the left and maximum degree $D_{\max}$ on the right, its expander code has a list decoding radius $\rho N = (\frac{1}{2} + \Omega(1/D_{\max}) )d$ and list size $N^{O(1)}$. Here $d$ is the distance of the code. 

Furthermore, if $D_{\max} \leq 1.1 D_R$, $\eps\le 1/4$ and $\alpha/\eps \le 0.1$, $\rho N$ is strictly larger than the Johnson bound $r$ of binary codes with distance $d=\frac{\alpha}{2\eps}\cdot N$.
\end{theorem}

We remark that, the Johnson bound $r=d/2 + \Theta(d^2/N)$ when $d$ is small. While we did not attempt to optimize the constant hidden in the $\Omega$ notation of $\rho = (\frac{1}{2} + \Omega(1/D_{\max}) )d$, we show that roughly $\frac{1}{D_R} \ge \frac{\alpha}{4\eps}$ in Section~\ref{sec:dist_expander}. When the expander is also almost regular on the right, e.g., $D_{\max} \leq 1.1 D_R$, this bound is better than the Johnson bound with $d=\frac{\alpha}{2\eps}\cdot N$ and a small ratio $\alpha/\eps$. The second condition would follow from a large average right-degree $D_R$ (equivalently, a small $M/N$ or a large code-rate $1-M/N$). In particular, this applies to the graph we construct in Section~\ref{sec:bound_dist}, which has distance arbitrarily close to $\frac{\alpha}{2\eps}\cdot N$. 

One intriguing question is to design efficient list-decoding algorithms for expander codes. Since these algorithms would also immediately improve all our results of unique decoding, we leave this as a future direction. 

\paragraph{New combinatorial properties of expander graphs.} Our distance bounds and decoding algorithms make extensive use of a new size-expansion tradeoff for bipartite expander graphs, which we establish in this paper. Specifically, we show that one can always trade the expansion for larger subsets in such a graph. In particular, given any $(\alpha N, (1-\eps)D)$-expander, we prove in Section~\ref{sec:dist_expander} that this graph is also roughly a $(k \alpha N, (1-k\eps)D)$-expander for any $k \ge 1$, provided that $k \alpha N \le N$. This size-expansion tradeoff is potentially of independent interest. For example, besides the applications in our distance bounds and decoding algorithms, we also use it to show a relation between the three basic parameters $(\alpha, \eps, D_R)$ of bipartite expanders. Roughly, we always have $\frac{\alpha}{\eps} \le \frac{4}{D_R}$ (see Fact~\ref{fact:relations_parameters} for a formal statement). On the other hand, using a random graph one can show the existence of $(\alpha N, (1-\eps)D)$-expanders such that roughly $\frac{\alpha}{\eps} \ge \frac{1}{e D_R}$ (see Proposition~\ref{prop:parameterrelation}). Thus our upper bound is tight up to a constant factor. 

\subsection{Related Work}


Sipser and Spielman's definition in \cite{SS96} is actually more general, and is a variant of Tanner codes \cite{Tanner81} based on expanders. Basically, the code requires all symbols in the neighbor set of a right vertex (in some fixed order) to be a codeword from an inner linear code $\mathcal{C}_0$. The expander code studied here is the most popular and well studied case, where the inner code consists of all strings with even weight. Instead of vertex expansion, the expander based Tanner codes are analyzed based on edge expansion, a related concept which has also been well studied in both mathematics and computer science \cite{LPS88,AlonChung1988}. We note that the distance of Tanner codes depends heavily on the inner code $\mathcal{C}_0$, and is thus generally incomparable to the distance of our code. To the best of our knowledge, the best bound on the distance of expander codes based on vertex expansion of bipartite expanders, as studied in this paper, was $2(1-\eps) \cdot \alpha N$.

As mentioned before, expander codes are closely related to low-density parity-check (LDPC) codes introduced by Gallager~\cite{Gallager1963}, where the bipartite graph associated with the parity checks has bounded degree on the right but is not necessary an expander. There is a long line of research on random LDPC codes against \emph{random errors} (see \cite{richardson2001the,Shok,ADS12} and the references therein). While a random LDPC code is an expander code with high probability, our results are incomparable with those of random LDPC codes. This is because first, we consider expander codes defined by arbitrary expanders, while many results on random LDPC codes use more properties than the expansion, such as the girth of the underlying graph that can be deduced from random graphs. Second, we consider adversarial errors, while many results on random LDPC codes \cite{richardson2001the,ADS12} consider random errors or memoryless channels. 

In the context of list decoding, the work of RonZewi-Wootters-Zemor \cite{ron2021linear} studied the problem of erasure list-decoding of expander codes, based on algebraic expansion properties (i.e., eigenvalues of the corresponding adjacency matrix).

In the past few decades, a great amount of research has been devoted to expander graphs, leading to a plethora of new results. We refer the reader to the survey by Hoory, Linial,  and Wigderson \cite{Hoory2006ExpanderGA} for an overview. Specifically, giving explicit constructions of bipartite expander graphs for expander codes has been a challenge. In particular, Kahale \cite{Kahale} showed that general Ramanujan graphs \cite{LPS88} (with the minimum 2nd largest absolute eigenvalue among all $D$-regular graphs) cannot provide vertex expansion more than half of the degree, which is the threshold required to give expander codes. After decades of efforts, explicit constructions satisfying the requirements of expander codes have been provided in \cite{AlonCapalbo,CRVW} separately.

\input{tech_overview}

\paragraph{Organization.} The rest of this paper is organized as follows.
In Section \ref{sec:preli}, we describe some basic notation, terms, definitions and   useful theorems from previous work.
In Section \ref{sec:dist_expander}, we show our improved distance bound for expander codes, and prove it is tight in general.
In Section \ref{sec:find}, we establish new properties of the algorithm which can find a super set of corruptions.
In Section \ref{sec:guess_flip}, we provide our main unique decoding algorithm.
In Section \ref{sec:improve_decoding}, we provide our improved unique decoding algorithm for smaller $\eps$.
In Section \ref{sec:list_decode}, we show our list-decoding result.
Finally, we conclude in Section \ref{sec:open} with some open questions.
Appendix~\ref{appen:proof_expansion} contains some relatively standard materials omitted in the main body.

%% file: tech_overview.tex
\subsection{Technique Overview}
Let $\mathcal{C}$ be an expander code defined by an $(\alpha N, (1-\eps)D)$ expander. Our techniques for the improved distance bound and decoding algorithms are based on the combination of the following three ingredients, together with a new idea of guessing expansions:
\begin{itemize}
    \item A new size-expansion tradeoff for arbitrary bipartite expander graphs, which we establish in this paper.
    \item A procedure of finding possible corruptions in \cite{Viderman13b}, which we slightly adapt and establish new properties.
    \item A procedure of flipping bits in the corrupted word to reduce the number of errors, introduced in \cite{SS96}.
\end{itemize}
We first briefly explain each ingredient.
\paragraph{The size-expansion tradeoff.} As mentioned before, we show that any $(\alpha N, (1-\eps)D)$-expander is also roughly a $(k \alpha N, (1-k\eps)D)$-expander for any $k \ge 1$. To prove this, assume for the sake of contradiction that there is a left subset $S$ with size $k \alpha N$ that has smaller expansion. This then implies that there are many \emph{collisions} (two different vertices on the left connected to the same vertex on the right) in the neighbor set of $S$, i.e., more than $k \eps D \cdot k \alpha N=k^2 \alpha \eps ND$ collisions. Now we pick a \emph{random} subset $T \subseteq S$ with size $\alpha N$, then each previous collision will remain with probability roughly $1/k^2$. By linearity of expectation, more than $\alpha \eps ND$ collisions are expected to remain in the neighbor set of $T$, thus implying the expansion of $T$ is smaller than $(1-\eps)D \cdot \alpha N$. This contradicts the expander property.

This convenient size-expansion tradeoff is used extensively in our bounds and algorithms. In fact, by using linear programming, we can get a better size-expansion tradeoff for $k \geq \frac{1}{2\eps}$, which we use in our result of list decoding expander codes.

\paragraph{The procedure of finding possible corruptions.} Viderman \cite{Viderman13b} introduced the following procedure for finding possible corruptions. Maintain a set $L$ of left vertices, a set $R$ of right vertices and a fixed threshold $h$. Start with $R$ being all the unsatisfied parity checks, then iteratively add left vertices with at least $h$ neighbors in $R$ to $L$, and their neighbors to $R$. Viderman showed that if the number of corruptions is not too large, then when this process ends, $L$ will be a super set of all corruptions and the size of $L$ is at most $\alpha N$. Therefore, one can treat $L$ as a set of erasures and decode from there.

In \cite{Viderman13b}, Viderman used sophisticated inequalities to analyze this procedure. In this paper, we show that the process has the following property.

\paragraph{Property (*):} If $h =(1-2 \Delta)D$ such that any subset $S$ of corrupted vertices has expansion at least $(1-\Delta)D |S|$, then all corruptions will be contained in $L$. Furthermore, we can assume without loss of generality that the set of corrupted vertices is added to $L$ before any other vertex. \\

This allows us to simplify the analysis in \cite{Viderman13b} and combine with our size-expansion tradeoff.

\paragraph{The procedure of flipping bits.} Sipser and Spielman \cite{SS96} introduced a procedure to flip bits in the corrupted word. Again, the idea is to set a threshold $h$, and flip every bit which has at least $h$ wrong parity checks in its neighbors. Sipser and Spielman showed that when $\eps < 1/4$ and the number of corruptions is not too large, this procedure will reduce the number of errors by a constant factor each time. Thus one only needs to run it for $O(\log N)$ times to correct all errors.

\paragraph{Our approaches.} We now describe how to combine these ingredients to get our bounds and algorithms. For the distance lower bound, it suffices to choose $k$ such that $1-k\eps > 1/2$. Then a standard analysis as in \cite{SS96} shows the distance of the code is at least $k \alpha N$. Thus, we can set $k \approx \frac{1}{2\eps}$ so that the distance is roughly at least $\frac{\alpha}{2 \eps} N$. A subtle point here is that it is not a prior clear that we can choose $k \approx \frac{1}{2\eps}$, since it may be that $k \alpha N =\frac{\alpha}{2 \eps} N > N$, and no left subset can have size larger than $N$. However, we again use the size-expansion tradeoff to show that this cannot happen. In particular, we show $\frac{\alpha}{\eps} \leq \frac{4}{D_R}$ (recall $D_R$ is the average degree on the right), and thus we can always set $k \approx \frac{1}{2\eps}$. Section~\ref{sec:bound_dist} gives a construction which shows this bound is almost tight.


Next we describe our decoding algorithms.

\paragraph{Unique decoding  for $\eps < 1/4$.} Our algorithm here is based on the following crucial observation. Let $F$ denote the set of corrupted vertices any time during the execution of the algorithm, and assume $|\Gamma(F)|=(1-\gamma)D|F|$, where $\Gamma(F)$ denotes the neighbor set of $F$. If $\gamma$ is large, or equivalently $|\Gamma(F)|$ is small, then the procedure of finding possible corruptions works well. This is because intuitively, the number of vertices added to $L$ will be proportional to $|\Gamma(F)|$, and thus $|L|$ will be small. On the other hand, if $\gamma$ is small, or equivalently $|\Gamma(F)|$ is large, then the procedure of flipping bits works well. This is because intuitively, the procedure of flipping bits works better when the expansion property is better.  

Hence, we can combine both procedures and set a threshold for $\gamma$. If $\gamma$ is larger than this threshold, we use the procedure of finding possible corruptions; otherwise we use the procedure of flipping bits.\ However, we don't know $\gamma$.\ Thus in our algorithm we guess $\gamma$, and for each possible value of $\gamma$ we apply the corresponding strategy. This is a bit like list-decoding, where we get a small list of possible codewords, from which we can find the correct codeword by checking the Hamming distance to the corrupted word. Note that the procedure of finding possible corruptions always returns a possible codeword; while to get a codeword from the procedure of flipping bits, we need to apply it for a constant number of times, until the number of errors is small enough so that we can easily correct all errors using any known algorithm. Thus we also need to guess $\gamma$ for a constant number of times.

Using these ideas, we show that Algorithm~\ref{alg:MultipleGuessWithFlips} can correct $(1-\eps)\alpha N$ errors for any constant $\eps < 1/4$. Now, we can improve this by combining with our size-expansion tradeoff. Specifically, for any constant $\eps < 1/4$ we can choose any $k \ge 1$ such that $k \eps < 1/4$. This implies that a modified algorithm can actually correct $(1-k \eps)k \alpha N$ errors. Setting $k \approx \frac{1}{4 \eps}$ gives us an algorithm that can correct roughly $\frac{3 \alpha}{16 \eps} N$ errors. 

For the running time, each time we guess $\gamma$, we know $\gamma=1-\frac{|\Gamma(F)|}{D|F|}$ with $|\Gamma(F)| \in [M]$ and $|F| \in [N]$. Thus a naive enumeration will result in $O(M N)=O(N^2)$ possible values. Since we need to guess $\gamma$ for a constant number of times, this will lead to a polynomial running time. However, instead we can enumerate $\gamma$ from $\{0, \eta, 2\eta, \ldots, \lceil \frac{1}{\eta} \rceil \eta\}$ for a small enough constant $\eta>0$. This reduces the running time to linear time, at the price of decreasing the relative decoding radius by an arbitrarily small constant. Finally, we remark that this algorithm
can be executed in logarithmic time on a linear number of parallel processors, since its main ingredients from \cite{SS96,Viderman13b} have parallel versions in logarithmic time.

\paragraph{Unique decoding  for smaller $\eps$.} When $\eps$ is even smaller, e.g., $\eps < 1/8$, our algorithm uses the procedure of finding possible corruptions, together with property (*) we established. Let $F$ denote the set of corrupted vertices in the received word. To use property (*), we need to find a $\Delta$ such that for any $S \subseteq F$, $S$ has expansion at least $(1-\Delta)D |S|$. Then we can set the threshold $h=(1-2 \Delta)D$. In \cite{Viderman13b}, one assumes $|F| \leq \alpha N$ and thus it is enough to set $\Delta = \eps$. However, our goal here is to correct more than $\alpha N$ errors, thus this choice of $\Delta$ no longer works. Instead, we use our size-expansion tradeoff to show that if $|\Gamma(F)|=(1-\gamma)D|F|$, then when $S \subseteq F$ and $|S| \ge \alpha N$, we always roughly have $|\Gamma(S)| \geq (1-\sqrt{\frac{\gamma |F| \eps}{\alpha N}})D |S|$, thus we can set $\Delta=\max\{\sqrt{\frac{\gamma |F| \eps}{\alpha N}}, \eps\}$.

However, again we don't know $\gamma$ and $|F|$. Thus we apply the same trick as before, and guess both quantities. This leads to Algorithm~\ref{alg:decoding_small_eps}. Since we have two possible cases ($\Delta=\sqrt{\frac{\gamma |F| \eps}{\alpha N}}$ or $\Delta =\eps$), we get two different decoding radii for different ranges of $\eps$. The running time is polynomial if we use the naive enumeration of $\gamma$ and $|F|$, but can be made linear by using a similar sparse enumeration as we discussed before.  

\paragraph{List decoding radius.}
Recall that our goal is to show that given any $y \in \mathbf{F}_2^N$, there is a list of at most $N^{O(1)}$ codewords within distance $\rho N = (\frac{1}{2} + \Omega(1/D_{\max}) )d $ to $y$.
Our analysis modifies the double counting argument that is used to show the Johnson bound.
The modification is by using the special structure of expander codes.

In more details, suppose the list of $L$ codewords within distance $\rho N$ to $y$, is $\{C_1, \ldots, C_L\}$.
Let $\tau_i$ be the number of codewords in the list which have their $i$-bit different from $y$.
We focus on counting the number $T$ of ``triples'' $(i, j_1, j_2)$, where the pair of codewords $(C_{j_1}, C_{j_2})$ are different in their $i$-th bit. 
Since the code has distance $d = \delta N$, we know $T \geq {L \choose 2}\delta N$.
We also know $T = \sum_{i\in [N]}\tau_i(L-\tau_i)$.
The key observation in our analysis is that for expander codes,  $\{\tau_i, i\in [N]\}$ have a large deviation.
Specifically, we call $\tau_i$ \emph{heavy} if $\tau_i \geq \frac{0.9}{D_{max}}L$, and show that the summation of heavy $\tau_i$'s is $\Theta(NL)$.
By using this observation, we manage to get a better upper bound for $T$ than that in the proof of the Johnson Bound in certain situations, which in turn yields a better list-decoding radius.



%% file: preli.tex
\section{Preliminary}\label{sec:preli}
We will use $1\{\mathcal{E}\} \in \{0,1\}$ to denote the indicator variable of a event $\mathcal{E}$. Moreover, we use $C$ and $c$ to denote different constants in various proofs of this paper.

\paragraph{Basic definitions from graph theory.} Given a graph $G$, we use $V(G)$ to denote its vertex set and $E(G)$ to denote its edge set. Given a bipartite graph $G$, we use $V_L(G)$ and $V_R(G
)$ to denote the left hand side and right hand side of the bipartite graph separately. When $G$ is clear, we simplify them as $V_L$ and $V_R$. Moreover, we fix two notations $N := |V_L|$ and $M := |V_R|$. 

For any subset $S \subseteq V_L \cup V_R$, we always use $\Gamma(S)$ to denote its neighbor set in $G$. If a vertex $v \in \Gamma(S)$ is connected to $S$ by exactly one edge, we call $v$ a unique neighbor of $S$ and use $\Gamma^1(S)$ to denote the set of all unique neighbors of $S$. 

In this work, we consider bipartite graphs that are regular on the left hand side. Thus we use $D$ to denote the \emph{regular} degree in $V_L$ and $D_R$ to denote the \emph{average} degree in $V_R$. 
Since $N=|V_L|$ and $M=|V_R|$, we have $N \cdot D = M \cdot D_R$. Moreover, we will use $D_{\max}$ to denote the \emph{maximum} degree in $G$, which would be the maximum degree in $V_R$ given $M<N$. 

A bipartite graph $G$ is an $(\alpha N, (1-\eps)D)$-expander if and only if for any left subset $S$ of size at most $ \alpha N $, its neighbor set $\Gamma(S)$ has size  $ \geq (1-\eps)D \cdot |S|$. For convenience, we call $\frac{|\Gamma(S)|}{|S|}$ the expansion of $S$ and say $G$ satisfies $(\alpha N, (1-\eps)D)$ expansion if and only if it is an $(\alpha N, (1-\eps)D)$-expander. Throughout this work, we assume that $D$ and $D_R$ are constants. Since we are interested in expanders with $\eps<1/2$ and $N>M$, we always assume $D > 2$ and $D_R > 3$.


\paragraph{Basic definitions from coding theory.}
We recall several notations from coding theory and define expander codes formally.
\begin{definition}
An $(N, k, d)$ binary error correcting code $\mathcal{C}$ is a set of codewords contained in $\mathbf{F}_2^N$, with $|\mathcal{C}|  = 2^k$ such that $\forall C_1, C_2 \in \mathcal{C}$, the Hamming distance between $C_1$ and $C_2$ is at least $d$. Moreover we call $k/N$ the rate of $\mathcal{C}$.

A linear code is a code whose codewords form a linear subspace of  $\mathbf{F}_2^N$.
\end{definition}
One fact about linear codes is that the distance of a linear code is equal to the minimum weight of non-zero codeword in it. 
The decoding radius of a decoding algorithm of $\mathcal{C}$ refers to the largest number of errors that the algorithm can correct. 

\begin{definition}[Expander Codes \cite{SS96}]
Given an $(\alpha N, (1-\eps)D)$ expander graph $G$ with $M$ right vertices, the expander code defined by $G$ is $\mathcal{C} \subseteq \mathbf{F}_2^N$ such that
\[
\mathcal{C} = \{ C\mid  \forall i\in [M], \sum_{j \in \Gamma(i)} C_j = 0\},
\]
where the addition is over the field $\mathbf{F}_2$.
\end{definition}

Given the definition of expander codes, we know its rate is $1-M/N$ and its distance is the minimum weight of non-zero codewords in $\mathcal{C}$.

\begin{remark}
The original definition of expander codes in \cite{SS96} is more general, where each vertex on the right represents some linear constraints on the codeword bits corresponding to its neighbors. In this paper, we only consider the most popular and well studied case where each vertex on the right represents a parity check.
\end{remark}

We use the following results of decoding for expander codes, from \cite{Viderman13b}.
\begin{theorem}[\cite{Viderman13b}]
\label{thm:vidermandecfromerasures}
Let $G$ be an $(\alpha N, (\frac{1}{2}+\xi)D)$ expander with $\xi>0$. For the expander code defined by $G$, there is a linear-time algorithm that can correct $ \alpha N$ erasures.

\end{theorem}

\begin{theorem}[\cite{Viderman13b}]
\label{thm:vidermandec}
Let $G$ be an $(\alpha N, (1-\eps)D)$ expander for $\eps<1/3$. For the expander code defined by $G$, there is a linear-time algorithm that can correct $\frac{1-3\eps}{1-2\eps} \lfloor \alpha N \rfloor $ errors.
\end{theorem}

%% file: distance.tex
\section{Improved Distance of Expander Codes}\label{sec:dist_expander}
Let $G$ be an $(\alpha N, (1-\epsilon) D)$ expander and $\mathcal{C}$ be the corresponding expander code. We show that when $\eps<1/2$, the distance of $C$ is roughly $\frac{1}{2\eps} \alpha N$.

\begin{theorem}\label{thm:dist_expander}
Let $G$ be an $(\alpha N, (1-\eps)D)$ bipartite expander. The distance of the expander code defined by $G$ is at least $\frac{\alpha}{2\eps} \cdot N - O(1/\eps)$.
\end{theorem}

In Section~\ref{sec:bound_dist}, we provide a construction of expander codes to show the above bound $\frac{\alpha}{2\eps} \cdot N$ is almost tight in general.

\begin{theorem}\label{thm:tight_distance}
Given any constants $0< \eps < 1/2$, $\eta>0$, there exist constants $D$ and $\alpha>0$, such that for infinitely many $N$, there exist $(\alpha N, (1-\eps-\eta)D)$-expanders with $M \in [N/2,2N/3]$ where (1) the rate of the expander code is in $[1/3,1/2]$;  and (2) the distance of the expander code is at most $\frac{\alpha}{2\eps} \cdot N$.
\end{theorem}

\begin{remark}
While the graphs we construct in Theorem~\ref{thm:tight_distance} are not strictly regular on the right, they are ``almost regular" in $V_R$, i.e., $D_{max} \leq 1.1 D_R$.
\end{remark}

To prove Theorem~\ref{thm:dist_expander}, we start with the following lemma which gives a tradeoff between the two parameters $\alpha$ and $\eps$. This is one of our main technical lemmas, and the proof is deferred to Section~\ref{sec:proof_lem_expansion}.  

\begin{lemma}\label{lem:expansion_larger}
For any $k>1$ and any left subset $S$ of size $k \alpha N$, we have
\begin{itemize}
    \item $|\Gamma(S)| \ge (1-k \eps) D \cdot k \alpha N - O(\eps D \cdot k^2)$.
    \item $|\Gamma(S)| \ge  (1-\frac{2k\eps - 1}{3-2/k}) \frac{k}{2} \cdot D \alpha N - O(D_R \cdot D)$ (which is better than the 1st bound for $k>1/2\eps$).
\end{itemize}
\end{lemma}


In particular, the first bound will be extensively used in our decoding algorithms, which shows an $(\alpha N, (1-\epsilon) D)$-expander is also roughly a $(k \alpha N, (1-k \epsilon) D)$-expander for any $k>1$. While this bound is extremely useful for $k \le 1/2\eps$, we will use the second one for larger $k$ to improve the list-decoding radius upon the standard Johnson bound.



Using the above lemma, we first prove the following facts in an expander graph.

\begin{fact}\label{fact:relations_parameters}
Let $G$ be an $(\alpha N, (1-\eps)D)$-expander with left regular degree $D$ and right average degree $D_R$. We always have
\begin{enumerate}
    \item $\eps \ge 1/D$.
    \item $\frac{\alpha}{4\eps} \le 1/D_R + O(\frac{1}{D \alpha} \cdot M/N^2)$.
\end{enumerate}
\end{fact}

\begin{proof}
To prove the first fact, let us consider the smallest non-trivial cycle $C$ in the expander graph $G$. First of all, we observe that $|C|=O(\log |V|)$. To show this, we consider the argument to bound the girth of a graph. Let us fix a vertex $v$ and consider the BFS tree with root $v$. The BFS procedure finds a non-trivial cycle when it finds a vertex in the 2nd time. Since $G$ is $D$-regular in $V_L$, up to depth $2 \log_{D-1} M$, 
the BFS procedure will find a non-trivial cycle. Then $|\Gamma(C \cap V_L)| \le D \cdot |C \cap V_L|-|C \cap V_R| = (D -1) \cdot |C \cap V_L|$.

For the second fact, consider a left subset $S$ of size $2M/D$ in $V_L$. Since $M \le N$ and $D \ge 2$, such an $S$ always exists. Since $|\Gamma(S)| \le M$, we always have
\[
(1 - \eps \cdot \frac{|S|}{\alpha N}) \cdot D \cdot |S| - O\left( \eps D \cdot (\frac{2M/D}{\alpha N})^2 \right) \le M.
\]
This implies
\begin{align*}
(1 - \eps \cdot \frac{2M/D}{\alpha N}) \cdot D \cdot 2M/D - O(\frac{\eps}{D \alpha^2} \cdot M^2/N^2) & \le M\\
(1 - \eps \cdot \frac{2}{\alpha D_R}) \cdot 2M - O(\frac{\eps}{D \alpha^2} \cdot M^2/N^2) & \le M \qquad (\text{recall } N D=M D_R)\\
M - O(\frac{\eps}{D \alpha^2} \cdot M^2/N^2) & \le \frac{4\eps}{\alpha D_R} \cdot M.
\end{align*}
So we have $\frac{1}{D_R} \ge \frac{\alpha}{4 \eps} - O(\frac{1}{D \alpha} \cdot M^2/N^2)$, or equivalently, $\frac{\alpha}{4\eps} \le 1/D_R + O(\frac{1}{D \alpha} \cdot M/N^2)$.
\end{proof}

We can now prove Theorem~\ref{thm:dist_expander}. 

\begin{proofof}{Theorem~\ref{thm:dist_expander}}
Suppose the claim is false, then there exists a non-zero codeword with Hamming weight the same as the distance. Consider any non-zero codeword $z$ of Hamming weight $\frac{\alpha}{2\eps} \cdot N - C/\eps$ for a sufficiently large constant $C$. Notice that the parameter $\frac{\alpha}{2 \eps} \le 2/D_R <1$ from Fact~\ref{fact:relations_parameters}. 

Let $S \subset [N]$ denote the entries in $z$ that are 1. By Lemma~\ref{lem:expansion_larger}, $\Gamma(S) \ge (1 - \frac{|S|}{\alpha N} \cdot \eps)D \cdot |S| - O(\eps D \cdot (\frac{|S|}{\alpha N})^2)$. For a sufficiently large constant $C$, this is at least
\[
(\frac{1}{2} + \frac{C/\eps}{\alpha N} \cdot \eps) D \cdot |S| - O(\eps D \cdot (\frac{|S|}{\alpha N})^2) > \frac{1}{2} \cdot D |S|.
\]
This implies the existence of unique neighbors in $\Gamma(S)$. Thus $z$ is not a valid codeword, which contradicts our assumption.
\end{proofof}




\subsection{Distance Upper Bound of Expander Codes}\label{sec:bound_dist}
In this section we prove Theorem~\ref{thm:tight_distance}. Given $\eta$ and $\eps$, let $D \ge \frac{1}{\eps \cdot \eta^2}$ be a constant  such that there exists a family of degree-$D$ Ramanujan graphs \cite{LPS88}  whose 2nd largest absolute value of eigenvalues of the adjacency matrix is $\lambda \le 2\sqrt{D-1}$. In this proof, when the graph $H$ is clear, we use $e(A,B)$ to denote the number of distinct edges between $A$ and $B$. At the same time, we state the following version of the expander mixing lemma by Alon and Chung \cite{AlonChung1988}.

\begin{lemma}[Lemma 2.3 in \cite{AlonChung1988}]\label{lem:expander_mixing}
Let $H=(V,E)$ be an expander with degree $D$, where the second largest absolute value of eigenvalues of the adjacency matrix is $\lambda$. Then for any subset $A \subset V$, $e(A,A)$, the number of edges inside $A$, is bounded by
\[
\left| e(A,A) - \frac{D}{2|V|} \cdot |A|^2 \right| \le \lambda/2 \cdot \sqrt{|A| \cdot (|V|-|A|)}.
\]
\end{lemma}

We now construct an $(\alpha N, (1-\eps-\eta)D)$-expander graph with $N+M$ vertices by putting together two disjoint graphs $G_0$ and $G_1$. For $G_0$, we first choose a Ramanujan graph $H$ with degree $D$ and size $N' = \frac{\alpha}{2\eps} \cdot N$ in the family for a sufficiently small $\alpha$. Then we construct $G_0$ as the vertex-edge bipartite graph corresponding to $H$. Namely $V_L(G_0)=V(H)$ and $V_R(G_0)=E(H)$ such that $(v,e) \in E(G_0)$ if and only if $v \in V(H)$ is a vertex in the edge $e \in E(H)$ of $H$. Notice that $G_0$ has left degree $D$.
\begin{claim}\label{claim:explicit_construction}
The bipartite graph $G_0$ constructed above is an $(\alpha N, (1-\eps-\eta)D)$-expander.
\end{claim}
\begin{proof}
For any $S \subseteq V_L(G_0)$, $|\Gamma(S)|$ is the number of distinct edges connected to $S \subset V(H)$ in $H$, i.e., $e(S,V(H))$ in the Ramanujan graph $H$. We rewrite $e(S,V(H))=e(S,S)+e(S,\overline{S})$. Since $2 e (S,S)+e(S,\overline{S})=D \cdot |S|$, we upper bound $e(S,S)$ by the expander mixing lemma:
\[
e(S,S) \le \frac{D}{2|V(H)|} \cdot |S|^2 + \lambda/2 \cdot \sqrt{|S| \cdot (|V(H)|-|S|)} \le D \cdot |S|  \left (\frac{|S|}{2|V(H)|} + \frac{\lambda}{2D} \cdot \sqrt{\frac{|V(H)-|S|}{|S|}} \right).
\]
Since $|S| \le \alpha N$, $|V(H)|=\frac{\alpha}{2\eps} \cdot N$ and $\lambda/D \le \frac{2}{\sqrt{D}} \le 2 \eta \sqrt{\eps}$, we have $e(S,S) \le (\eps + \eta) D \cdot |S|$ and $e(S,V(H)) = D \cdot |S| - e(S,S) \ge (1-\eps-\eta) \cdot D |S|$.
\end{proof}

Then we construct $G_1$ as a random regular bipartite graph with $|V_L(G_1)|=N_1=N-N'$, $|V_R(G_1)|=M_1=M-DN'/2$, regular left degree $D$ and regular right degree $D_R=N_1 \cdot D/M_1$. Since we can choose $\alpha$ to be sufficiently small and $M \in [N/2,2N/3]$, such an integer $D_R$ exists. Furthermore, a random bipartite graph with such parameters satisfies $(\alpha N, (1-\eps)D)$ expansion with high probability for a small $\alpha$. For completeness we show this calculation in Appendix~\ref{appen:proof_expansion} and assume this property in the rest of this proof. 
Overall, because both $G_0$ and $G_1$ satisfy $(\alpha N, (1-\eps-\eta)D)$ expansion, $G= G_0 \cup G_1$ is an $(\alpha N, (1-\eps-\eta)D)$ expander. Moreover, $G$ is almost regular since $G_0$ is small and $G_1$ is regular on both sides.

Finally, consider a cedeword that is all $1$ in $V_L(G_0)$, and $0$ everywhere else. It satisfies all parity checks since the right degree of $V_R(G_0)$ is 2. Moreover, its weight is $\frac{\alpha}{2\eps} N$, and thus the distance of the corresponding expander code is at most $\frac{\alpha}{2\eps} N$.


\subsection{Proof of Lemma~\ref{lem:expansion_larger} and Its Generalization}\label{sec:proof_lem_expansion}
We prove the first lower bound $|\Gamma(S)| \ge (1-k \eps) D \cdot |S| - O(D \cdot k^2)$ by a probabilistic argument.
Suppose $|\Gamma(S)|$ is small. Then we consider a random subset $T$ of size $\alpha N$ in $S$ and upper bound 
\[
\big| \Gamma(T) \big| \le D \cdot |T| - \left(|S| \cdot D - \big |\Gamma(S) \big| \right) \cdot \frac{|T| \cdot (|T|-1)}{|S| \cdot (|S| - 1)}.
\]
The reason is that consider any neighbor $u$ of $S$ with more than $1$ neighbors in $S$, say $u$ has $d_S(u)$ neighbors in $S$ which are $v_1,\ldots,v_{d_S(u)}$. Since \[
1\{u \in \Gamma(T)\} \le \sum_{i=1}^{d_S(u)} 1\{v_i \in T\} - \sum_{i=2}^{d_S(u)} 1\{v_1\in T\} \cdot 1\{v_i\in T\}
\]
Then we take expectation (over $T$) on both sides:
\begin{equation}\label{eq:single_vertex}
\E_T\left[1\{u \in \Gamma(T)\}\right] \le d_S(u) \cdot \frac{|T|}{|S|} - (d_S(u)-1) \cdot \frac{|T| \cdot (|T|-1)}{|S| \cdot (|S|-1)}.
\end{equation}
At the same time, we know 
\begin{equation}\label{eq:two_sums}
    \sum_{u \in \Gamma(S)} d_S(u)=D \cdot |S| \text{ and } \sum_{u \in \Gamma(S)} (d_S(u)-1) = D \cdot |S| - |\Gamma(S)|
\end{equation}

Then we consider the summations over $u\in \Gamma(S)$ on the two sides of \eqref{eq:single_vertex}: By linearity of expectation, it becomes
\begin{align*}
\E_T\left[ \big| \Gamma(T) \big| \right] & \le \sum_u d_S(u) \cdot \frac{|T|}{|S|} - \sum_u (d_S(u)-1) \cdot \frac{|T| \cdot (|T|-1)}{|S| \cdot (|S|-1)} \tag{plug the two summations of \eqref{eq:two_sums}}\\
 & =  D \cdot |T| - \bigg(|S| \cdot D - \big| \Gamma(S) \big| \bigg) \cdot \frac{|T| \cdot (|T|-1)}{|S| \cdot (|S|-1)} \\
& = |T| \cdot D \left( 1 - \bigg( 1 - \frac{\big| \Gamma(S) \big|}{D \cdot |S|} \bigg) \cdot \frac{ |T|-1}{ |S| - 1} \right) .
\end{align*}
On the other hand, this is at least $|T| \cdot D (1-\eps)$ by the expander property. So we have
\[
1-\eps \le 1 - \bigg( 1 - \frac{\big| \Gamma(S) \big|}{D \cdot |S|} \bigg) \cdot \frac{ |T|-1}{ |S| - 1} \quad  \Leftrightarrow \quad \eps / \frac{ |T|-1}{ |S| - 1} \ge 1 - \frac{\big| \Gamma(S) \big|}{D \cdot |S|}.
\]
This gives
\[
\frac{\big| \Gamma(S) \big| }{D \cdot |S|} \ge 1 - \eps \cdot \left( k + \frac{k-1}{\alpha N - 1} \right).
\]
We rewrite it to obtain
\[
\big| \Gamma(S) \big| \ge (1-\eps k) \cdot D |S| - \eps \frac{(k-1)}{\alpha N-1} \cdot D |S| \ge (1-\eps k) \cdot D |S| - 2 \eps D k^2.
\]
\paragraph{Generalization.}
Next we consider an alternative way to compute $\E[\Gamma(T)]$. The main motivation is to prove a better bound  than the above one for $k>1/2\eps$.

Let us fix $S$ of size $k \alpha N$ and consider $\Gamma(S)$. Since the total degree of $S$ is $D \cdot k \alpha N$, let $\beta_j \cdot D \alpha N$ denote the number of vertices in $\Gamma(S)$ with exactly $j$ neighbors in $S$. By the definition,
\[
\Gamma(S)=(\beta_1 + \cdots + \beta_k) \cdot D \alpha N.
\]
Moreover, By summing up the degrees, we have
\[
\beta_1 + 2 \beta_2 + \cdots + D_R \cdot \beta_{D_R} = k.
\]
Now we consider 
\begin{equation}\label{eq:exp_neighbors}
\E[\Gamma(T)]=\sum_{i \in \Gamma(S)} \Pr_{T \sim {S \choose \alpha N}}[T \cap \Gamma(i) \neq \emptyset],
\end{equation}
which is at least $(1-\eps)D \alpha N$ from the property of expansion.

For each $i \in \Gamma(S)$ with exactly $j$ edges to $S$, 
\[
\Pr_{T}[T \cap \Gamma(i) \neq \emptyset] = 1 - \Pr_{T}[T \cap \Gamma(i) = \emptyset] = 1 - \frac{(|S|-|T|) \cdot (|S|-|T|-1) \cdots (|S|-|T|- j + 1)}{|S| \cdot (|S|-1) \cdots (|S|-j+1)}.
\]
Since we assume $k,D,D_R,\alpha,\eps=\Theta(1)$ and $|S|,|T|=\Theta(N)$, we simplify this probability to
\[
\Pr_T[T \cap \Gamma(i) \neq \emptyset]=1 - (\frac{|S|-|T|}{|S|})^j + \frac{O(j)}{|S|} = 1 - (1- \frac{1}{k})^j + O(\frac{D_R}{k \alpha N})
\]
and omit the error term $O(\frac{D_R}{k \alpha N})$ for ease of exposition. 
Plugging this into Eq\eqref{eq:exp_neighbors}, we have the inequality
\[
\sum_{j=1}^{D_R} \left[ 1 - (1-\frac{1}{k})^j \right] \cdot \beta_j \ge (1-\eps).
\]
To lower bound $\Gamma(S)$, we rewrite all constraints as a linear programming:
\begin{align}
    & \min \beta_1 + \cdots + \beta_{D_R} \notag \\
    \text{subject to } & \beta_1 + 2 \cdot \beta_2 + \cdots D_R \cdot \beta_{D_R} = k \label{eq:sum_deg}\\
    & \sum_{j=1}^{D_R} \left[ 1 - (1-\frac{1}{k})^j \right] \cdot \beta_j \ge (1-\eps) \label{eq:expected_neighbors}\\
    & \beta_j \ge 0, \quad \forall j. \notag
\end{align}

Next we prove that (1) the minimum is achieved by $\beta^*$ with at most two non-zero entries; (2) more importantly, if $\beta^*$ has exactly two non-zero entries, they must be adjacent. We could consider the dual of the above linear program
\begin{align}
    & \max k \cdot x_1 + (1-\eps) \cdot x_2 \notag\\
    \text{subject to } & j \cdot x_1 + [1-(1-1/k)^j] \cdot x_2 \le 1 \quad \forall j=1,\ldots,D_R, \label{eq:dual_constraint}\\
    & x_2 \ge 0. \notag
\end{align}
The key property is that those coefficients $\left[ 1 - (1-\frac{1}{k})^j \right]$ in constraints \eqref{eq:expected_neighbors} and \eqref{eq:dual_constraint} constitute a strictly concave curve. Namely,
for any $j$,
\begin{equation}\label{ineq:concave}
    \left[ 1 - (1-\frac{1}{k})^j \right] - \left[ 1 - (1-\frac{1}{k})^{j-1} \right] > \left[ 1 - (1-\frac{1}{k})^{j+1} \right] - \left[ 1 - (1-\frac{1}{k})^{j} \right].
\end{equation}
Inequality \eqref{ineq:concave} is  true for any $j>1$ since $1+(1-1/k)^2 > 2 (1-1/k)$. 
For contradiction, if $\beta^*$ is supported on three entries say $\ell_1<\ell_2<\ell_3$, we have $j \cdot x_1 + [1-(1-1/k)^j] \cdot x_2=1$ for $j=\ell_1,\ell_2,\ell_3$ by the slackness of linear programming. However, the two equations for $j=\ell_1, \ell_3$ indicate $\ell_2 \cdot x_1 + [1-(1-1/k)^{\ell_2}] \cdot x_2 > 1$. While their linear combination with coefficients $\frac{\ell_3-\ell_2}{\ell_3-\ell_1}$ and $\frac{\ell_2-\ell_1}{\ell_3-\ell_1}$ equals $1$ on RHS from these two equations, this combination on LHS  is strictly less than $\ell_2 \cdot x_1 + [1-(1-1/k)^{\ell_2}] \cdot x_2$ in the equation of $j=\ell_2$ by $x_2 \ge 0$ and the concavity of $[1-(1-1/k)^j]$. Similarly, if $\beta^*$ is supported on two non-adjacent entries say $\ell_1<\ell_2$, we have two equations for $j=\ell_1$ and $j=\ell_2$ separately. However, the solution of $(x_1,x_2)$ which satisfies these two equations violates other constraints in the dual --- one can show $(\ell_1+1) \cdot x_1 + [1-(1-1/k)^{\ell_1+1}] \cdot x_2 > 1$ by the same argument again.

Now we prove the 2nd lower bound. Let us consider the dual probram where $x_1$ and $x_2$ are determined by \eqref{eq:dual_constraint} with $j=2$ and $j=3$:
\begin{align*}
    2 x_1 + (2/k - 1/k^2) x_2 & = 1, \\
    3 x_1 + (3/k - 3/k^2 + 1/k^3) x_2 & = 1.
\end{align*}
We get $x_1=\frac{2-1/k-k}{3-2/k}$ and $x_2=\frac{k^2}{3-2/k}$ such  that the objective value is $k \cdot x_1 + (1-\eps)x_2= \frac{2k-1-\eps k^3}{3-2/k}=\frac{k}{2} (1-\frac{2k\eps - 1}{3-2/k})$ (we omit the error $k \cdot O(\frac{D_R}{k \alpha N})$). One can verify this pair of $(x_1,x_2)$ is feasible: \eqref{eq:dual_constraint} is true for $j=1$ since $x_1+\frac{1}{k} x_2=\frac{2-1/k}{3-2/k} \le 1$ given $k>1$. \eqref{eq:dual_constraint} is also true for $j>3$ by the concavity again. So this shows $\min \beta_1 + \cdots + \beta_{D_R}$ of the original LP is at least $\frac{k}{2} (1-\frac{2k\eps - 1}{3-2/k})$.

In fact, the first lower bound is obtained from the dual where $x_1$ and $x_2$ are determined by  \eqref{eq:dual_constraint} with $j=1$ and $j=2$. Finally we remark that when $1-\eps \in \bigg[ \frac{k}{j+1} \cdot [1-(1-1/k)^{j+1}], \frac{k}{j} \cdot [1-(1-1/k)^j] \bigg]$, the two non-zero entries of $\beta^*$ in the primal are $\beta^*_j$ and $\beta^*_{j+1}$. This indicates that $\frac{k}{2} (1-\frac{2k\eps - 1}{3-2/k})$ is the optimum value of the linear program when $1-\eps \in [1 - \frac{1}{k} + \frac{1}{3k^2}, 1 - \frac{1}{2k}]$, or equivalently, $1/2\eps \le k \le \frac{2/3}{1 - \sqrt{1 - 4\eps/3}}=\frac{1+\sqrt{1-4\eps/3}}{2\eps}$.

%% file: find.tex
\section{Decoding from Erasures, and Finding Possible Corruptions}\label{sec:find}
First, we show that by combining Lemma \ref{lem:expansion_larger} and Theorem \ref{thm:vidermandecfromerasures}, we can also get a stronger result for decoding from erasures.  
\begin{theorem}
\label{thm:decfromerasures}
Consider an expander code defined by an $(\alpha N, (1-\eps)D)$ expander $G$.
For every  $\xi > 0$, there is a linear-time algorithm that   corrects $\frac{1-\xi}{2\eps}\alpha N$ erasures.
\end{theorem}
\begin{proof}
By Lemma \ref{lem:expansion_larger}, for any $1 \leq k \leq \frac{1}{\alpha}$ the expander is also a $(k\alpha N, (1-k\eps)D-O_\eps(\frac{1}{ N}))$ expander.
Thus if $1 - k\eps - O_\eps(1/N) >1/2 +  \xi'$ for a $\xi' > 0$, then by Theorem \ref{thm:vidermandecfromerasures}, one can decode from $k\alpha N$ erasures, using the same algorithm.
This means $k $ can be as large as $ \frac{1-\xi}{2\eps} $ for any   $\xi > 0$. Notice that if $ \frac{1-\xi}{2\eps} \leq 1 $ the theorem is trivially true by Theorem \ref{thm:vidermandecfromerasures}, and $ \frac{1-\xi}{2\eps} \leq \frac{1}{\alpha} $ by Fact~\ref{fact:relations_parameters}.
\end{proof}

Next, we provide a simple algorithm to find a super set of the corruptions, which is adapted from a similar algorithm in \cite{Viderman13b}.
Let $G$ be an $(\alpha N, (1-\eps)D)$ expander with $N$ left vertices $M$ right vertices and left degree $D$. 
Let $\mathcal{C}$ be an expander code defined by $G$. 
The input $y$ is a corrupted message of a codeword $C_0 \in \mathcal{C}$.
Let $F$ be the set of corruptions in $y$ compared to $C_0$. 
We use Algorithm \ref{alg:FindingErasure} to find a super set of $F$ given certain parameters.
\begin{algorithm}[h]
\caption{The basic algorithm finding a super set of corruptions}\label{alg:FindingErasure}
\begin{algorithmic}[1]
\Function{Find}{$y \in \mathbf{F}_2^N$ and $\Delta \in \mathbf{R}$}
\State $L \gets \emptyset$ 
\State $R \gets \{\text{unsatisfied parity checks of } y\}$
\State $h \gets (1-2\Delta)D$
\While{$\exists i \in V_L \setminus L$ s.t. $|\Gamma(i) \cap R| \ge h  $}
\State $L \gets L \cup \{i\}$
\State $R \gets R \cup \Gamma(i)$
\EndWhile
\State \Return $L$
\EndFunction
\end{algorithmic}
\end{algorithm}


By a similar proof to that of proposition 4.3 in \cite{Viderman13b}, we   have the following properties. 
\begin{lemma}
\label{lemma:all_errors_in_L}
If $|\Gamma^1(S)|  \geq (1-2\Delta)D |S|$ for any non-empty $S \subseteq F$, then $F$ is contained in $L$ after the while loop.
\end{lemma}
\begin{proof}
Suppose not, then let $B $ be $ F \setminus L$  after running the algorithm, $B\neq \emptyset$.
Since $B \subseteq F$, we have $|\Gamma^1(B)| \geq (1-2\Delta)D|B|$. 
So there is a vertex $u\in B$ such that  $u$ has at least $(1-2\Delta)D$ unique neighbors in $\Gamma(B)$.
We know that $|\Gamma(u) \cap R| < (1-2\Delta)D$, because otherwise $u$ should be added to $R$ then.
Thus there has to be a neighbor $v$ of $u$, such that $v$ is not in $R$ and is only connected to one vertex in $B$, which is $u$.
As $F\setminus B \subseteq L$,  we know $\Gamma(F\setminus B) \subseteq R$.
So $v$ connects to one vertex, i.e., $u$ in $F$. 
This is not possible since then $v$ has to be unsatisfied and thus it is already in $R$. 
\end{proof}

\begin{lemma}
\label{lemma:L'allsame}
In every iteration, if there are multiple vertices that can be added to $L$ and we choose one of them arbitrarily, then we always get the same $L$ after all the iterations.
\end{lemma}
\begin{proof}
Consider two different procedures where they choose different vertices to add to $L$ in their corresponding iterations.
Suppose that they get two different $L$, say $L_1$ for the first procedure and $ L_2$ for the second.
Without loss of generality assume $L_1 \setminus L_2 \neq \emptyset$.
Let $u$ be the first vertex in  $L_1 \setminus L_2$ that is added in procedure 1.
Then all the vertices in $L_1$ added before $u$, denoted by the set $A$, is also contained in $L_2$. Since vertices can only be added to the set $R$, 
for procedure 2 we should always have $|\Gamma(\{u\})\cap R| \geq h$ when $ A \subseteq L_2$ and $u \notin L_2$. 
Thus $u$ has to be added to $L_2$ in procedure 2.
This is a contradiction.
Therefore $L_1 = L_2$.
\end{proof}

\begin{lemma}
\label{lemma:FAddtoLFirst}
If $ |\Gamma^1(S)|  \geq (1-2\Delta)D |S| $ for any non-empty $S \subseteq F$,
then there exists a sequence of choices of the algorithm such that all the elements of $F$ can be added to $L$ in the first $|F|$ iterations.
\end{lemma}
\begin{proof}
We use inductions to show that in each of the first $|F|$ iterations, there exists an element in $F\setminus L$ which can be added to $L$.

In the first iteration, since $|\Gamma^1(F)| \geq (1-2\Delta)D |F|$, there exists $u\in F$ such that $|\Gamma(u) \cap \Gamma^1(F)| \ge h$ for $h=(1-2\Delta)D$. 
Notice that all these neighbors should be unsatisfied and thus in $R$.
So $u$ can be added to $L$.

Assume in each of the first $i-1<|F|$ iterations, the algorithm can  find a distinct element in $F$ to add to $L$.
In the $i$-th iteration, let $B = F\setminus L$.
Notice that $|F'| = |F|-(i-1)\geq 1$.
Hence $|\Gamma^1(F')|\geq (1-2\Delta)D|F'|$.
Thus there exists $u \in F'$ such that  $|\Gamma(u) \cap \Gamma^1(F')| \geq (1-2\Delta)D = h$. 
Notice that all these neighbors of $u$ in $\Gamma^1(F')$ have to be in $R$, since otherwise there is a neighbor not in $R$, which   also corresponds to an unsatisfied check, contradicting that all the unsatisfied checks are in $R$.
Hence $|\Gamma(u) \cap R | \geq (1-2\Delta)D = h$.
So $u$ can be added to $L$.
\end{proof}

The above lemmas imply that as long as $ |\Gamma^1(S)|  \geq (1-2\Delta)D |S| $ for any non-empty $S \subseteq F$, when analyzing Algorithm \ref{alg:FindingErasure}, we can assume without loss of generality that the algorithm first adds all corrupted bits into the set $L$.

%% file: guessingwithflips.tex
\section{Unique Decoding by Guessing Expansion with Iterative Flipping }\label{sec:guess_flip}
Let $\eps \in (0,  1/4)$ be an arbitrary constant in this section. We first show an algorithm which has a decoding radius $(1-\eps)\alpha N$. 
Then by using Lemma \ref{lem:expansion_larger}, we show that the algorithm achieves decoding radius approximately $\frac{3 \alpha}{16\eps} N$.

The basic idea of the algorithm is to guess the expansion of the set of corrupted entries in the algorithm, say $(1-\gamma)D$.
Assume we can correctly guess $\gamma$.
For the case of $\gamma \geq \frac{2}{3}\eps$, we use a procedure similar to \cite{Viderman13b} to find a super set of possible corruptions, and then decode from  erasures. In this case we show that $(1-\eps)\alpha N$ errors can be corrected.
For the case of $\gamma < \frac{2}{3}\eps$, we first consider the left subset which contains all vertices with at least $1-3\gamma$ unsatisfied checks, and show that this set contains (a constant fraction) more corrupted bits than correct bits. Thus we can flip all bits in this set and reduce the number of errors by a constant fraction. The algorithm then repeats this step for a constant number of times, until the number of errors is small enough, where we can apply an existing algorithm to correct the remaining errors.

We describe our algorithm in Algorithm~\ref{alg:MultipleGuessWithFlips} and then state our main result of this section.
\begin{algorithm}[H]
\caption{Decoding Algorithm for $\eps = 1/4- \beta$}\label{alg:MultipleGuessWithFlips}
\begin{algorithmic}[1]
\Function{Final Decoding}{$y \in \mathbf{F}_2^n, \alpha  \in \mathbf{R}, \eps \in \mathbf{R}$} //The main procedure.
 \State Enumerate $\gamma_1, \gamma_2, \ldots, \gamma_l$ where $\ell = O(\log_{1-O(\beta)}\frac{1}{3})$.
 For every $ i\in [\ell]$,   $\gamma_i$ is enumerated   from the set $\{ \eta, 2\eta, \ldots, \lceil\frac{1}{\eta} \rceil \eta \}$, where $\eta:=\beta/100$.
    \For{enumeration $\gamma_1, \gamma_2, \ldots, \gamma_\ell$}
        \State $C' \gets$ \Call{Decoding}{$y,  \gamma_1, \ldots,\gamma_{\ell}, \alpha, \eps$} 
        \If{$C'$ is a valid codeword  and the distance between $C'$ and $y$ is at most $(1-\eps) \alpha N$} 
        \State \Return $C'$
        \EndIf
    \EndFor
\EndFunction

\Function{Decoding}{$y \in \mathbf{F}_2^n$ and $(\gamma_1,\ldots,\gamma_{\ell}) \in \mathbf{R}^{\ell}, \alpha \in \mathbf{R}, \eps \in \mathbf{R}$}
    \State $z \gets y$
 \For{$i=1,\ldots,\ell$}
    \If{$\gamma_i  \ge 2\eps/3 + \eta $}
     \State $z \gets$ \Call{FixedFindAndDecode}{$z, \alpha, \eps$} 
     \State \Return $z$ 
    \Else
    \State Let $L_0$ denote all bits in $z$ with at least $(1-3 \gamma_i)D$ wrong parity checks
        \State Flip all the bits in $L_0$
    \EndIf
\EndFor
    \State Apply the decoding of Theorem \ref{thm:vidermandec} on $z$ and return the result.
\EndFunction

\Function{FixedFindAndDecode}{$y \in \mathbf{F}_2^N, \alpha \in \mathbf{R}, \eps \in \mathbf{R}$}
\State $L \gets $ \Call{Find}{$y, \eps$}, where \Call{Find}{} is from Algorithm \ref{alg:FindingErasure}.
\State Erase the symbols of $y$ that have indices in $L$ to get $y'$, and then apply the decoding from Theorem \ref{thm:decfromerasures}  on $y'$ to get a codeword $C'$.
\State \Return $C'$
\EndFunction

\end{algorithmic}
\end{algorithm}

\begin{theorem}
\label{thm:MultipleGuessWithFlips}
For every small constant $\beta > 0$, and every $\eps \leq 1/4 - \beta $, let $\mathcal{C}$ be an expander code defined by a $(\alpha N, (1-\eps)D)$ expander graph. There is a linear time decoding algorithm for $\mathcal{C}$ with decoding radius $(1-\eps ) \cdot \alpha N$.

\end{theorem}
 
To prove the theorem, we focus on the $i$-th iteration of Function \textsc{Decoding}, and show that we can make progress (either reducing the number of errors or decoding the original codeword) in this iteration.
Let $F_i$ denote the set of errors at the beginning of iteration $i$ and $\gamma(F_i) \in [0,\eps]$ be the parameter such that $| \Gamma(F_i) | = (1-\gamma(F_i)) \cdot D|F_i|$.

First we show Function \textsc{FixedFindAndDecode} will recover the codeword directly whenever $\gamma_i \ge \frac{2\eps}{3}+\eta$ and $\gamma(F_i) \ge \frac{2\eps}{3}$. 

\begin{claim}\label{clm:small_gamma}
If $|F_i| \leq (1-\eps  ) \cdot \alpha N$,   $\gamma_i \ge \frac{2\eps}{3} + \eta$, and $\gamma(F_i) \in [\gamma_i-\eta ,  \gamma_i )$, then Function \textsc{FixedFindAndDecode} in \textsc{Decoding} will return a valid codeword directly.
\end{claim} 


\begin{proof}
First notice that when $\gamma_i \ge \frac{2\eps}{3}+\eta$, this iteration of Function \textsc{Decoding} will go to Function \textsc{FixedFindAndDecode}. Let $\gamma:= \gamma(F_i)$. We prove that $L$ after \textsc{Find} has size at most $\alpha N$. 
Suppose not.
Since $|F_i| \leq (1-\eps  ) \cdot \alpha N$, by the expander property, for every nonempty $ F'\subseteq F_i,  |\Gamma(F')| \geq (1-\eps)D|F'|$, so by Lemma \ref{lemma:all_errors_in_L},   after \textsc{Find}, $L$ covers all the errors. 
Consider the moment $|L|=\alpha N$. 
Without loss of generality, we assume $F_i \subseteq L$ (otherwise we can adjust the order of vertices added to $L$ by Lemma \ref{lemma:L'allsame}).

Then we have
\[
(1-\eps) D \alpha N \le | \Gamma(L) | \le  (1-\gamma)D \cdot |F_i| + 2 \eps D (\alpha N - |F_i|),
\]
because the expansion of $F_i$ is $(1-\gamma)D \cdot |F_i|$ and when adding any vertex in $L\setminus F_i$ to $L$, $R$ increases by at most $2\eps D$.
So
\[
(1-\eps) \alpha N   \le (1-\gamma) \cdot |F_i| + 2 \eps (\alpha N - |F_i|).
\]

As $\gamma  \leq \eps$ and $\eps \leq 1/4$, $1-\gamma - 2 \eps > 0$. This implies $|F_i| \ge \frac{1-3\eps}{1-\gamma - 2 \eps} \cdot \alpha N $. 
Since $\gamma  \geq \gamma_i -\eta   \geq \frac{2\eps}{3}  $, we have   $|F_i| \geq \frac{1-3\eps}{1-8\eps/3  } \alpha N  $. 
When  $\eps \leq 1/4 - \beta$ and $\eta$ is small enough, one can check that $\frac{1-3\eps}{1-8\eps/3   }  > 1-\eps   $ always holds.
It is contradicting the assumption that $|F_i| \leq (1-\eps  ) \alpha N$.

As $L \supseteq F_i$ and is of size at most $\alpha N$, the algorithm can correct all the errors using $L$ and $z$ given $\eps <1/4-\beta$ to Theorem~\ref{thm:decfromerasures}.
\end{proof}

Next we discuss the case where  $\gamma_i<2\eps/3+\eta$, which will result in the Function \textsc{Decoding} finding the set $L_0$ and flipping all the bits in $L_0$. We show that this will reduce the number of errors by a constant fraction.

\begin{claim}\label{clm:large_gamma}
If $|F_i| \leq (1-\eps  ) \alpha N$,  $\gamma_i  < \frac{2\eps}{3}+\eta$, and $\gamma(F_i) \in [\gamma_i-\eta, \gamma_i )$, then flipping $L_0$ will reduce the distance between $z$ and the correct codeword by at least $ \beta$ fraction.
\end{claim}
\begin{proof}
Let $\gamma:=\gamma(F_i)$ and $N':=\frac{(1+ 3\eta)|F_i|}{(1-\eps) \alpha}$. We show that $|F_i \cup L_0| < \alpha N' $. 
To prove it,  assume $|F_i \cup L_0|=\alpha N'$, i.e., we only take $\alpha N' - |F_i|$ elements from $L_0\setminus F_i$, consider these elements together with elements in $F_i$.
As $|F_i| \leq (1-\eps)\alpha N$, $\alpha N' \leq (1+3\eta)\alpha N$.
By Lemma \ref{lem:expansion_larger}, $(1-(1+3\eta)\eps)D \alpha N' - O(\eps D (1+3\eta)^2)   \le |\Gamma(F_i \cup L_0)|$.
Also notice that $|\Gamma(F_i)| = (1-\gamma)D |F_i|$ and adding each element of $L_0\setminus F_i$ to $L_0$ contributes at most $3\gamma_i D$ to  $|\Gamma(F_i \cup L_0)|$.
So
\[
(1-(1+3\eta)\eps - O_\eps(1/N))D \alpha N'   \le |\Gamma(F_i \cup L_0)| \le (1-\gamma)D |F_i| + 3 \gamma_i D \cdot (\alpha N' - |F_i|).
\]
This implies $|F_i| \ge \frac{1-(1+3\eta)\eps- O_\eps(1/N) - 3 \gamma_i}{1-\gamma - 3 \gamma_i} \cdot \alpha N'$.
As $  \gamma_i \leq \gamma + \eta $, this is $\geq\frac{1-(1+3\eta)\eps - O_\eps(1/N) -  3 \gamma-3\eta}{1-4 \gamma - 3\eta} \cdot \alpha N'$ 
which is minimized when $\gamma  = 0$. 
Thus  $|F_i|\geq\frac{(1-(1+3\eta)\eps - O_\eps(1/N)-3\eta )}{1-3\eta} \alpha N' = (1- \eps -\frac{6\eta}{1-3\eta}\eps - O_\eps(1/N))\alpha N'$.
But we know that $|F_i| = \frac{1-\eps}{1+3\eta}\alpha N' = (1-\eps-\frac{3\eta-3\eps \eta}{1+3\eta})\alpha N' $. 
This is a contradiction.

Let $\eps' = (1+3\eta)\eps + O_\eps(1/N)$.
Since $|F_i\cup L_0| <\alpha N'$, now again we have
\begin{equation}\label{eq:F_and_L_0}
   (1-\eps') \cdot D |F_i \cup L_0| \le |\Gamma(F_i \cup L_0) | \le (1-\gamma)D|F_i| + 3 \gamma_i D|L_0 \setminus F_i|.
\end{equation}
As $ |F_i \cup L_0| = |L_0 \setminus F_i| + |F_i| $,  
this gives
\[
(1-\eps'-3 \gamma_i)  \cdot |L_0\setminus F_i| \le (\eps'-\gamma) \cdot |F_i|.
\]

Now we consider $F_i \setminus L_0$. Each variable in $F_i \setminus L_0$ contributes at most $(1-3\gamma_i) D$ unique neighbors. And the total number of unique neighbors is at least $(1-2\gamma)D \cdot |F_i|$. So the size of $F_i\setminus L_0$ is at most $\frac{2\gamma}{3\gamma_i} \cdot |F_i|$.


Finally we prove $|F_i\setminus L_0|+|L_0\setminus F_i| \le (1- \beta) \cdot|F_i|$. From the above bounds on $|F_i\setminus L_0|$ and $|L_0\setminus F_i|$, it is enough to show
\[ 
\frac{2\gamma}{3\gamma_i}  + \frac{\eps' -\gamma}{1-\eps' - 3 \gamma_i}< 1 -   \beta.
\]
As $\gamma \in [\gamma_i -\eta, \gamma_i)$,
the L.H.S. is at most
\[
\frac{2  }{3 }  + \frac{\eps' - \gamma_i  }{1-\eps' - 3 \gamma_i }.
\]
We know  $\eps \leq 1/4 - \beta$. So $\eps' \leq (1/4-\beta)(1+3\eta) + O_\eps(1/N) < 1/4- 0.9\beta $ when $\eta \le \beta/100$. 
Also we know $\gamma_i \in [0, \frac{2}{3}\eps+\eta)$. 
So when $\gamma_i = 0$, $\eps' = 1/4-0.9\beta$, the L.H.S. is at most   $1- \beta$.
\end{proof}

\begin{proofof}{Theorem \ref{thm:MultipleGuessWithFlips}}
The decoding algorithm is Algorithm \ref{alg:MultipleGuessWithFlips}.
Given the parameter $\beta$, we first set up a large enough constant $\ell=\Theta(1/\beta)$.
Then we apply the algorithm on an input corrupted codeword, using parameter $\ell$. 
The key point is that in the enumerations of the $\gamma_i$'s, one sequence $(\gamma_i)_{i\in [l]}$ provides a good approximation of the actual expansion parameters, i.e. $\forall i\in [l]$ in the $i$-th iteration,  $\gamma(F_i)\in [ \gamma_i-\eta, \gamma_i)$.
Now we consider every iteration in this sequence.
If $\gamma_i \geq 2\eps/3 +\eta $, then by Claim \ref{clm:small_gamma}, the algorithm returns the correct codeword.
If $\gamma_i < 2\eps/3 + \eta$, then by  Claim \ref{clm:large_gamma}, the number of errors can be reduced by $ \beta $ fraction. 
So in the worst case, when $\ell \geq \log_{1-O(\beta)}\frac{1}{3}$, the number of errors can be reduced to at most $\alpha N / 3$ in a constant number of iterations.
Finally the algorithm applies the decoding algorithm from Theorem \ref{thm:vidermandec}, which corrects the remaining errors.

The running time of Algorithm \ref{alg:MultipleGuessWithFlips} is linear, since $\ell = O(1)$ and our enumeration for each $\gamma_i$ takes constant time. The procedures \textsf{FixedFindAndDecode} and the decoding from Theorem \ref{thm:vidermandec} both run in linear time as well.

\end{proofof}

By using Theorem \ref{thm:MultipleGuessWithFlips} and Lemma \ref{lem:expansion_larger} we can get the following result.
\begin{theorem}\label{thm:MultipleGuessWithFlipsFinal}
For every constants $ \eps \in (0, \frac{1}{4}), \eta>0$, 
if $\mathcal{C}$ is an expander code defined by an $(\alpha N, (1-\eps)D)$ expander,
 then there is a linear time decoding algorithm for $\mathcal{C}$ with decoding radius $(\frac{3 \alpha}{16 \eps} - \eta) N$. 
\end{theorem}

\begin{algorithm}[H]
\caption{Decoding Algorithm for $\eps<1/4$ with larger decoding radius}\label{alg:MultipleGuessWithFlipsLarger}
\begin{algorithmic}[1]
\Function{Final Decoding For Large Radius}{$y\in \mathbf{F}_2^N, \alpha \in \mathbf{R}, \eps \in \mathbf{R}$}
\State Let $k = (1/4-\beta - O_\eps(1/N))/\eps$, with $\beta$ being a small enough constant.
\State Let $z \gets$  \Call{Final Decoding}{$y, k\alpha, k \eps + O_\eps(1/N)$}, where  \Call{Final Decoding}{} is from Algorithm~\ref{alg:MultipleGuessWithFlips}
\State \Return $z$.
\EndFunction
\end{algorithmic}
\end{algorithm}
\begin{proof}
Consider Algorithm \ref{alg:MultipleGuessWithFlipsLarger}.
By Lemma \ref{lem:expansion_larger}, the expander graph is also a $ (k\alpha N, (1-k\eps)D - O_\eps(1/N)) $ expander for $k\geq 1$.
If $k$ satisfies $(1-k\eps)D - O_\eps(1/N) \leq 1/4 - \beta$ for a small constant $\beta$, then by Theorem \ref{thm:MultipleGuessWithFlips}, there is a decoding algorithm with radius $(1 - k\eps - O_\eps(1/N))k\alpha N$.
When $k = (1/4-\beta - O_\eps(1/N))/\eps$, this is maximized to be $(\frac{3}{16\eps} - O_{\eps}(\beta))\alpha N$. 
Here we take $\beta=O(\eta)$ to be small enough such that $k\geq 1$ and the decoding radius becomes $(\frac{3 \alpha}{16 \eps} - \eta) N$.
The running time is linear by Theorem~\ref{thm:MultipleGuessWithFlips}.
\end{proof}

%% file: Improved_unique_decoding.tex
\section{Improved Unique Decoding  for  $\eps \leq 1/8$}\label{sec:improve_decoding}
In this section we provide two decoding algorithms for the case of $\eps \leq 1/8$, with better decoding radius. The first one is a polynomial time decoding algorithm. The second one is a linear time algorithm with slightly worse parameters.

Consider the expander code based on an $(\alpha N, (1-\eps)D)$ bipartite expander.
For the case of $\eps \le 1/8$ we provide an efficient algorithm,  i.e. Algorithm~\ref{alg:decoding_small_eps}, to decode from more errors.
It is again by guessing the correct expansion of the set of corrupted entries. 
We first give a polynomial time decoding algorithm, then we modify it to give a linear time decoding algorithm.

\subsection{Polynomial time decoding}

\begin{algorithm}[H]
\caption{Decoding Algorithm for $\eps\le 1/8$}\label{alg:decoding_small_eps}
\begin{algorithmic}[1]
\Function{Decoding}{$y \in \mathbf{F}_2^N,\epsilon \in \mathbf{R},\alpha \in \mathbf{R}$}
	\For {every integer $i \in [1, N]$, let $x$ satisfy $x \alpha N = i$} 
	\For {every integer $j \in [1, M]$, let $\gamma$ satisfy $(1-\gamma) \cdot D \cdot i = j $}
	    \If{$\gamma x \ge \eps$}
	        \State $\Delta \gets \sqrt{\gamma x \eps} + \frac{c}{N}$ with $c:=c(\eps)$ being a large enough constant.
	    \Else
	        \State $\Delta  \gets \eps + \frac{c}{N}$.
	    \EndIf
	    \State $L \gets $ \Call{Find}{$y \in \mathbf{F}_2^n$, $\Delta$}
	    \State 
	    Erase the symbols of $y$ that have indices in $L$ to get $y'$, and then apply the decoding from Theorem \ref{thm:decfromerasures}  on $y'$ to get a codeword $C'$.
	    \State \Return $C'$ if the  distance between $C'$ and $y$ is $\le \frac{1-2\eps}{4 \eps}\alpha N $. 
	\EndFor
    \EndFor
\EndFunction
\end{algorithmic}
\end{algorithm}

\begin{theorem}
\label{thm:guessexpansion}
For every constant $\eps \in (0, 1/8]$,  if $\mathcal{C}$ is an expander code defined by an $(\alpha N, (1-\eps)D)$ expander, then there is a polynomial time  decoding algorithm for $\mathcal{C}$ with decoding  radius $\frac{\sqrt{2}-1}{2\epsilon} \alpha N -O_{\eps}(1) $ when $\eps< \frac{3-2\sqrt{2}}{2}$ and decoding  radius $\frac{1-2\eps}{4\eps} \alpha N - O_{\eps}(1) $ when $\eps \ge \frac{3-2\sqrt{2}}{2}$.


\end{theorem}

We prove the correctness of Algorithm~\ref{alg:decoding_small_eps} and the theorem in the rest of this section. Again $F$ always denotes the set of corrupted entries. Since we enumerate both $x \alpha N = i$   and $(1-\gamma)D x\alpha N = j$  over all possible values. 
One pair of them corresponds to  the correct size of $F$ and   the correct expansion of $F$, i.e., $|F| = x \alpha N$ and $|\Gamma(F)|=(1-\gamma) \cdot D |F|$.
Now we only consider this pair $(x, \gamma)$  in the following analysis.

First of all, we can bound the expansion of all subsets in $F$.
\begin{claim}\label{clm:worst_expansion}
Our choice of $\Delta$ always satisfies that
\[
\forall F' \subseteq F, |\Gamma(F')| \ge (1-\Delta) \cdot D |F'|.
\]
\end{claim}
\begin{proof}
Let $F' \subseteq F$ be an arbitrary non-empty set, and $|F'| = x' \cdot \alpha N$. 

If $x'>1$, then assume $|\Gamma(F')| = (1-\beta)D x'\alpha N$.
we consider the collisions   in $\Gamma(F)$ and $\Gamma(F')$.
~By collision we mean that given an arbitrary order of the edges, if one edge in this order has its right endpoint the same as any other edge prior to it, then this is called a collision. 
Note that the total number of collisions for edges with left endpoints in $F'$ is at most the total number of collisions for edges with left endpoints in $F$, because a collision  in $\Gamma(F')$ is also a collision in $\Gamma(F)$. Thus  $$\beta x' \leq \gamma x.$$
Also, since $F'$ has size $x' \cdot \alpha N$, by Lemma \ref{lem:expansion_larger} we have $|\Gamma(F')| \geq (1-x'\eps)D  x'\alpha N - O_\eps(1)$.
So $ \beta \leq x'\eps + O_\eps(1/N)$.
Hence  $ \beta (\beta - O_\eps(1/N))/\eps \leq \gamma x $.
Thus $ \beta \leq \sqrt{\gamma x \eps } + O_\eps(1/N) $ and  $|\Gamma(F')| = (1-\beta)D |F'| \geq (1-\sqrt{\gamma x \eps}) D |F'| - O_\eps(1)$.
When $\gamma x\geq \eps$, the algorithm sets $\Delta = \sqrt{\gamma x \eps} + c/N$.
So $|\Gamma(F')| \geq (1-\Delta) D |F'|$, when $c$ is large enough.
When $\gamma x < \eps$, the algorithm set $\Delta  = \eps + c/N $.
Notice that $\sqrt{\gamma x \eps} \leq \eps$.
Hence again $|\Gamma(F')| \geq (1-\Delta) D |F'|$, when $c$ is large enough.

If $x' < 1$, then again we have two cases.
When $ \gamma  x \ge \eps  $,  we know $\Delta = \sqrt{\gamma x \eps} \geq \eps $.
So by expansion, $ |\Gamma(F')| \ge (1-\eps)D |F'| \geq (1-\Delta)D |F'| $.
When  $ \gamma x < \eps   $, the algorithm sets $\Delta = \eps+ c/N$.
So $ |\Gamma(F')| \ge (1-\eps)D |F'| \geq (1-\Delta)D |F'| $.

\end{proof}

Given the guarantee in Claim~\ref{clm:worst_expansion}, one can show that $L$ contains all the errors.
\begin{claim}
\label{claim:FsubseteqL}
After step 9 in Algorithm~\ref{alg:decoding_small_eps}, we have $F\subseteq L$.
\end{claim}

\begin{proof}
By Claim \ref{clm:worst_expansion}, $\forall F' \subseteq F, |\Gamma(F')| \ge (1-\Delta) \cdot D |F'|$.
So $\forall F' \subseteq F, |\Gamma^1(F')| \ge (1-2\Delta) \cdot D |F'|$, by noticing that $1-2\Delta > 0$ in our setting.
By Lemma \ref{lemma:all_errors_in_L}, we know $F\subseteq L$ after \textsc{Find}.
\end{proof}

Then we calculate the decoding radius and the size of $L$.

\begin{claim}
\label{claim:guessExpansionLbound1}
For the branch which sets $\Delta=\sqrt{\gamma x \eps } + \frac{c}{N}$, if $x  \leq \frac{\sqrt{2}-1}{2\eps} - O_\eps(1/N) $, then $|L| < \frac{1-2\eps}{2\eps} \alpha N$.
\end{claim}
\begin{proof}
Suppose after the iterations, $|L| \ge \frac{1-2\eps}{2\eps} \alpha N$.
By Claim \ref{clm:worst_expansion} and Lemma \ref{lemma:FAddtoLFirst},  we can consider an $L'$ which is constituted by first adding $F$ and then adding another $  \frac{1-2 \Delta }{2\eps} \alpha N -  x \alpha N$ elements.  
Let $\delta= \frac{|L'| - |F|}{\alpha N} =   \frac{1 -2  \Delta  }{2 \eps} - x$. 
Notice that $\frac{|L'|}{\alpha N} =  \frac{1 -2  \Delta  }{2 \eps} \leq \frac{1-2\eps}{2\eps}$.
We   show that   even having this $L'$ leads to a contradiction.

We notice that $\delta \geq 0$ and $x+ \delta \geq 1$. 
The reason is as follows.
Notice that we only need to consider the case $x\geq 1$, since otherwise $\gamma x < \eps$ and thus the algorithm should not go to this branch.
Notice that $\gamma \leq x\eps + O_\eps(1/N) $ by Lemma \ref{lem:expansion_larger}.
So $\delta = \frac{1 -2 (\sqrt{\gamma x \eps} + c/N ) }{2 \eps} - x   \geq  \frac{1}{2\eps} - 2x- O_\eps(1/N)$.
When $x \leq \frac{\sqrt{2}-1}{2\eps} - O_\eps(1/N) $ and $\eps \leq 1/8$, this is at least $0$. 
Hence $ x+ \delta \geq 1 $. 

Next notice that all the unsatisfied checks are in $\Gamma(F)$ and every element in $L'\setminus F$ contributes at most $2\Delta D$ vertices to $R$. 
Hence $|\Gamma(L')| \le |\Gamma(F)|+2\Delta D\cdot \delta \alpha N$. On the other hand, Lemma~\ref{lem:expansion_larger} implies $|\Gamma(L')| \ge (1-(x+\delta)\eps)D \cdot (x+\delta)\alpha N - O_\eps(1)$. Thus we have
\[
(1-(x+\delta)\eps) \cdot (x+\delta)\alpha N - O_\eps(1)\le (1-\gamma)x \alpha N + 2 \Delta \cdot \delta \alpha N \le  (1-\gamma)x \alpha N + 2 (\sqrt{\gamma x \eps}+ c/N ) \cdot \delta \alpha N..
\]
In the rest of this proof, we show that our choice of $\delta$ yields
\begin{equation}\label{eq:ineq1}
(1 - (x+\delta)\eps) \cdot (x+ \delta) - O_\eps(1/N) > (1-\gamma) x + 2 (\sqrt{\gamma x \eps} + c/N)   \cdot \delta. 
\end{equation}
This gives a contradiction. Towards that, we rewrite  inequality~\eqref{eq:ineq1} as
\[
0 > \eps \delta^2 + \left(2 \eps x - 1 + 2 (\sqrt{\gamma x \eps} +  c/N)  \right)\delta + \eps x^2 - \gamma x +  O_\eps(1/N) .
\]
When $(2 \eps x - 1 + 2 (\sqrt{\gamma x \eps} + c/N)  )^2 - 4 \eps \left( \eps x^2 -\gamma x +  O_\eps(1/N)   \right)>0$, the quadratic polynomial will be negative at $\delta =\frac{1 - 2 \eps x -2 (\sqrt{\gamma x \eps}+ c/N)  }{2\eps} = \frac{1-2\Delta}{2\eps} -x $. 
To verify this,   we set $z=\eps x$ and only need to verify that
\[
(2z - 1 + 2  \sqrt{\gamma z} )^2 - 4 z^2 + 4 \gamma z -  O_\eps(1/N) >0.
\]
This is equivalent to
\[
8 \gamma z + (8z -4 ) \sqrt{\gamma z} + 1 - 4z  - O_\eps(1/N) > 0 \Rightarrow 8(\sqrt{\gamma z} +\frac{2z-1}{4})^2  - O_\eps(1/N) + 1 - 4z - 8(\frac{2z-1}{4})^2 > 0.
\]
When $z=\eps x \leq \frac{\sqrt{2}-1}{2} - O_\eps(1/N)$, i.e. $x \leq \frac{\sqrt{2}-1}{2\eps} - O_\eps(1/N)$,
the residue $1 - 4z - 8(\frac{2z-1}{4})^2 -O_\eps(1/N) = 1/2-2z -2 z^2 - O_\eps(1/N)>0$. 
So the inequality holds.
\end{proof}

\begin{claim}
\label{claim:guessExpansionLbound2}
For the branch $\Delta=\eps$, if  $x \leq \frac{1-2\eps}{4\eps} - O_\eps(1/N)$, then $|L| < \frac{1-2\eps}{2\eps} \alpha N$.
\end{claim}
\begin{proof}

Suppose $|L| \geq \frac{1-2\eps}{2\eps} \alpha N$.
Consider $L'\subseteq L$ with $|L'| =  \frac{1-2\eps}{2\eps} \alpha N$.
Let $\delta = \frac{1-2\eps}{2\eps} - x$. 
Notice that $\delta \geq 0$ because    $x \leq  \frac{1-2\eps}{4\eps} - O_\eps(1/N)$.
Also $x+ \delta \geq 1$ since $\eps \leq 1/8$.
By Lemma \ref{lem:expansion_larger}, $|\Gamma(L')| \geq (1-(x+\delta)\eps)D |L'| - O_\eps(1)$.
By Lemma \ref{lemma:FAddtoLFirst} we can consider $L'$ as being constituted by first adding all elements in $F$ and then add another $\delta \alpha N$ elements by the algorithm.
Notice that all the unsatisfied checks are in $\Gamma(F)$, $|\Gamma(F)| \leq D|F|$, and every element in $L'\setminus F$ contributes at most $2\eps D$ vertices to $R$. 
Hence $|\Gamma(L')| \leq D|F| + 2\eps D \delta \alpha N$.
So we have
\[
(1-(x+\delta)\eps)D |L'| - O_\eps(1) \leq  |\Gamma(L')| \leq D|F| + 2\eps D \delta \alpha N
\]
Thus
\[
(1-(x+\delta)\eps) \cdot (x+\delta) - O_\eps(1/N) \le x + 2 \eps \delta.
\]
So this is equivalent to
\[
(1- 2\eps - \eps (x+\delta))(x+\delta)  - O_\eps(1/N)  \le (1-2\eps) x
\]
Recall that $\delta+x=\frac{1 -2 \eps}{2\eps}$. 
To get a contradiction, we only need 
\[
(1-2\eps) x < (1-2\eps)^2/4\eps - O(1/N).
\]
Namely $x \leq \frac{1-2\eps}{4\eps} - O_\eps(1/N)$.

\end{proof}

\begin{proofof}{Theorem \ref{thm:guessexpansion}}
One of our enumerations correctly predicts  $|F|$ and  $|\Gamma(F)|$.
Consider   Algorithm \ref{alg:decoding_small_eps} under this enumeration.
After the function $\textsc{Find}$, all the errors are in $L$ by Claim \ref{claim:FsubseteqL}. 

Now we bound $|L|$. 
We can pick the smaller bound of $x$ from Claim   \ref{claim:guessExpansionLbound1} and Claim   \ref{claim:guessExpansionLbound2}.
If $\eps < \frac{3-2\sqrt{2}}{2} $, then $\frac{\sqrt{2}-1}{2\eps} < \frac{1-2\eps}{4\eps}$. 
So by Claim   \ref{claim:guessExpansionLbound1} and Claim   \ref{claim:guessExpansionLbound2}  when $x \leq \frac{\sqrt{2}-1}{2\eps} - O_\eps(1/N) $ we have $|L| < \frac{1-2\eps}{2\eps}\alpha N  $.
If $\eps \in [\frac{3-2\sqrt{2}}{2}, 1/8]$, then $\frac{\sqrt{2}-1}{2\eps} \geq \frac{1-2\eps}{4\eps}$.
So by Claim   \ref{claim:guessExpansionLbound1} and Claim   \ref{claim:guessExpansionLbound2} , when 
$x \leq \frac{1-2\eps}{4\eps} - O_\eps(1/N) $, we have $|L| < \frac{1-2\eps}{2\eps} \alpha N$. 
Since the expander is an $(\alpha N, (1-\eps)D)$ expander,  by Theorem \ref{thm:decfromerasures}, one can correct all the errors efficiently using $L$ (as the set of erasures) and the corrupted codeword.

\end{proofof}


\subsection{Linear time decoding}
\begin{algorithm}[H]
\caption{Decoding Algorithm for $\eps\le 1/8$}\label{alg:decoding_small_eps_lineartime}
\begin{algorithmic}[1]
\Function{Decoding}{$y \in \mathbf{F}_2^N,\epsilon \in \mathbf{R},\alpha \in \mathbf{R},\eta' \in \mathbf{R}$}
\State Enumerate  $\tilde{\gamma}\tilde{x}$ from $\{ \eta , 2\eta , \ldots, \lceil\frac{1}{\eta } \rceil \eta  \}$, where $\eta=\Theta(\eps \cdot \eta')$.
	    \If{$\tilde{\gamma} \tilde{x}\ge \eps   $  }
	        \State $\Delta \gets \sqrt{\tilde{\gamma} \tilde{x} \eps}   + \eta $.
	    \Else
	        \State $\Delta  \gets \eps + 2\eta $.
	    \EndIf
	    \State $L \gets $ \Call{Find}{$y \in \mathbf{F}_2^n$, $\Delta$}
	    \State 
	    Erase the symbols of $y$ that have indices in $L$ to get $y'$, and then apply the decoding from Theorem \ref{thm:decfromerasures}  on $y'$ to get a codeword $C'$.
	    \State \Return $C'$ if distance between $C'$ and $y$ is $\le \frac{1-2\eps}{4 \eps}\alpha N $ where the distance comes from Theorem~\ref{thm:dist_expander}.
\EndFunction
\end{algorithmic}
\end{algorithm}

\begin{theorem}
\label{thm:guessexpansion_lineartime}
For all constants $\eps \in (0, 1/8], \eta' > 0$,  if $\mathcal{C}$ is an expander code defined by an $(\alpha N, (1-\eps)D)$ expander, then there is a linear time  decoding algorithm for $\mathcal{C}$ with decoding  radius $(\frac{\sqrt{2}-1}{2\epsilon}\alpha -  \eta' ) N  $ when $\eps< \frac{3-2\sqrt{2}}{2}$ and decoding  radius $(\frac{1-2\eps}{4\eps}\alpha -  \eta'  ) N   $ when $\eps \ge \frac{3-2\sqrt{2}}{2}$.

\end{theorem}

We prove the correctness of Algorithm~\ref{alg:decoding_small_eps_lineartime} and the theorem in the rest of this section. Again $F$ always denotes the set of corrupted entries, and assume $|F| = x\alpha N$, $|\Gamma(F)| = (1-\gamma)D|F|$.
Since we enumerate $\tilde{\gamma}\tilde{x}$ from a sequence with gap $\eta$, one  of them is such that $  \gamma x \in [\tilde{\gamma}\tilde{x}  , \tilde{\gamma}\tilde{x} + \eta]$.
Now we only consider this  enumeration  in the following analysis.

Next we can bound the expansion of all subsets in $F$.
\begin{claim}\label{clm:worst_expansion_lineartime}
Our choice of $\Delta$ always satisfies that
\[
\forall F' \subseteq F,   |\Gamma(F')| \ge (1-\Delta) \cdot D |F'|.
\]
\end{claim}
\begin{proof}
Let $F' \subseteq F$ be an arbitrary non-empty set, and $|F'| = x' \cdot \alpha N$.

If $x'>1$, then assume $|\Gamma(F')| = (1-\beta)D x'\alpha N$.
we consider the collisions   in $\Gamma(F)$ and $\Gamma(F')$.
~Recall that by collision we mean that given an arbitrary order of the edges, if one edge in this order has its right endpoint the same as any other edge prior to it, then this is called a collision. Note that the total number of collisions for edges with left endpoints in $F'$ is at most the total number of collisions for edges with left endpoints in $F$, because a collision  in $\Gamma(F')$ is also a collision in $\Gamma(F)$. Thus
$$\beta x' \leq \gamma x.$$
Also, since $F'$ has size $x' \cdot \alpha N$, by Lemma \ref{lem:expansion_larger} we have $|\Gamma(F')| \geq (1-x'\eps)D  x'\alpha N - O_\eps(1)$.
So $ \beta \leq x'\eps + O_\eps(1/N)$.
Hence  $ \beta (\beta - O_\eps(1/N))/\eps \leq \gamma x $.
Thus $ \beta \leq \sqrt{\gamma x \eps } + O_\eps(1/N) $ and  $|\Gamma(F')| = (1-\beta)D |F'| \geq (1-\sqrt{\gamma x \eps}) D |F'| - O_\eps(1)$.
When $ \tilde{\gamma} \tilde{x}\ge \eps $, the algorithm sets $\Delta = \sqrt{\tilde{\gamma} \tilde{x}\eps} + \eta$.
Notice that $  \sqrt{\gamma x\eps }  \leq \Delta - \eta/2$.
So $|\Gamma(F')| \geq (1-\Delta) D |F'|$.
When $ \tilde{\gamma} \tilde{x} < \eps$, the algorithm sets $\Delta  = \eps + 2\eta$.
Notice that $\sqrt{\gamma x \eps} \leq \eps + \eta$.
Hence again $|\Gamma(F')| \geq (1-\Delta) D |F'|$.

If $x' < 1$, then again we have two cases.
When $\tilde{\gamma} \tilde{x}\ge \eps  $,  we know $\Delta \geq \eps + \eta$.
So by expansion, $ |\Gamma(F')| \ge (1-\eps)D |F'| \geq (1-\Delta)D |F'| $.
When  $\tilde{\gamma} \tilde{x}< \eps   $, the algorithm sets $\Delta = \eps+2\eta$.
So $ |\Gamma(F')| \ge (1-\eps)D |F'| \geq (1-\Delta)D |F'| $.

\end{proof}

Given the guarantee in Claim~\ref{clm:worst_expansion_lineartime}, one can show that $L$ contains all the errors.
\begin{claim}
\label{claim:FsubseteqL_lineartime}
After step 9 in Algorithm~\ref{alg:decoding_small_eps_lineartime}, we have $F\subseteq L$.
\end{claim}

\begin{proof}
By Claim \ref{clm:worst_expansion_lineartime}, $\forall F' \subseteq F, |\Gamma(F')| \ge (1-\Delta) \cdot D |F'|$.
So $\forall F' \subseteq F, |\Gamma^1(F')| \ge (1-2\Delta) \cdot D |F'|$, since $(1-2\Delta) > 0$ in our setting.
By Lemma \ref{lemma:all_errors_in_L}, we know $F\subseteq L$ after \textsc{Find}.
\end{proof}

Then we calculate the decoding radius and the size of $L$.

\begin{claim}
\label{claim:guessExpansionLbound1_lineartime}
 
For the branch $\Delta=\sqrt{\tilde{\gamma}\tilde{x} \eps } + \eta  $, if $x \leq \frac{\sqrt{2}-1}{2\eps} - O(\eta/\eps) $, then $|L| < \frac{1-2\eps}{2\eps} \alpha N$.
\end{claim}
\begin{proof}
First of all, we will use the fact $\Delta \le \sqrt{\gamma x \epsilon} + \eta$ (since $\gamma x \in [\tilde{\gamma}\tilde{x},\tilde{\gamma}{\tilde{x}}+\eta]$ in the correct guessing) extensively in this proof. Now suppose after the iterations, $|L| \ge \frac{1-2\eps}{2\eps} \alpha N$.
By Claim \ref{clm:worst_expansion_lineartime} and Lemma \ref{lemma:FAddtoLFirst},  we can consider an $L'$ which is constituted by first adding $F$ and then adding another $  \frac{1-2 (\sqrt{\gamma x \eps} + \eta ) }{2\eps} \alpha N - x \alpha N$ elements.  
Let $\delta= \frac{|L| - |F|}{\alpha N} =   \frac{1 -2 (\sqrt{\gamma x \eps} + \eta)  }{2 \eps} - x$.
Notice that $|L'| =  \frac{1 -2 (\sqrt{\gamma x \eps} + \eta)  }{2 \eps} \leq \frac{1-2\eps}{2\eps}$.
We   show that   even having this $L'$ leads to a contradiction.

We notice that $\delta \geq 0,  x+ \delta \geq 1$. 
The reason is as follows.
First consider the case $x\geq 1$.
Notice that $\gamma \leq x\eps + O_\eps(1/N) $ by Lemma \ref{lem:expansion_larger}.
So $\delta = \frac{1 -2 (\sqrt{\gamma x \eps} + \eta) }{2 \eps} - x   \geq  \frac{1}{2\eps} - 2x- O_\eps(\eta)$.
when $x \leq \frac{\sqrt{2}-1}{2\eps} - O_\eps(\eta) $, $\eps \leq 1/8$, this is at least $0$. 
Also notice that $ x+ \delta \geq  \frac{1 }{2 \eps} - x  -\Theta_\eps(\eta) \geq 1 $ when $\eps \leq 1/8$.
Second if $x<1$, then $\tilde{\gamma}\tilde{x} \leq \gamma x < \eps $ and thus the algorithm should not go to this branch.

Next notice that all the unsatisfied checks are in $\Gamma(F)$ where $|\Gamma(F)| = (1-\gamma) D|F|$, and every element in $L'\setminus F$ contributes at most $2\Delta D$ vertices to $R$. 
Hence $|\Gamma(L')| \le |\Gamma(F)|+2\Delta D\cdot \delta \alpha N$. On the other hand, Lemma~\ref{lem:expansion_larger} implies $|\Gamma(L')| \ge (1-(x+\delta)\eps)D \cdot (x+\delta)\alpha N - O_\eps(1)$. Thus we have
\[
(1-(x+\delta)\eps) \cdot (x+\delta)\alpha N - O_\eps(1)\le (1-\gamma)x \alpha N + 2 \Delta \cdot \delta \alpha N \le  (1-\gamma)x \alpha N + 2 (\sqrt{\gamma x \eps} + \eta) \cdot \delta \alpha N..
\]
In the rest of this proof, we show that our choice of $\delta$ yields
\begin{equation}\label{eq:ineq1_lineartime}
(1 - (x+\delta)\eps) \cdot (x+ \delta) - O_\eps(1/N) > (1-\gamma) x + 2 (\sqrt{\gamma x \eps} + \eta)   \cdot \delta.
\end{equation}
This gives a contradiction. Towards that, we rewrite  inequality~\eqref{eq:ineq1_lineartime} as
\[
0 > \eps \delta^2 + \left(2 \eps x - 1 + 2 (\sqrt{\gamma x \eps} + \eta)  \right)\delta + \eps x^2 - \gamma x +  O_\eps(1/N) .
\]
When $(2 \eps x - 1 + 2 (\sqrt{\gamma x \eps} + \eta)  )^2 - 4 \eps \left( \eps x^2 -\gamma x +  O_\eps(1/N)   \right)>0$, the quadratic polynomial will be negative at $\delta =\frac{1 - 2 \eps x -2 (\sqrt{\gamma x \eps} + \eta)  }{2\eps} $. 
To verify this,   we set $z=\eps x$ and only need to guarantee that
\[
(2z - 1 + 2  \sqrt{\gamma z} )^2 - 4 z^2 + 4 \gamma z -  2(1 - 2\eps x - 2 \sqrt{\gamma x \epsilon)}\eta >0.
\]
This is equivalent to
\[
8 \gamma z + (8z -4 ) \sqrt{\gamma z} + 1 - 4z  - 2\eta > 0 \Rightarrow 8(\sqrt{\gamma z} +\frac{2z-1}{4})^2  - 2\eta + 1 - 4z - 8(\frac{2z-1}{4})^2 > 0.
\]
When $z=\eps x \leq \frac{\sqrt{2}-1}{2} - O(\eta)$ (namely $x \leq \frac{\sqrt{2}-1}{2\eps} - O(\eta/\epsilon)$),
the residue $1 - 4z - 8(\frac{2z-1}{4})^2 - 2 \eta = \frac{1}{2} -2z -2z^2 - 2 \eta >0$. 
So the inequality holds.
\end{proof}

\begin{claim}
\label{claim:guessExpansionLbound2_lineartime}
For the branch $\Delta=\eps + 2\eta$, if  $x \leq \frac{1-2\eps}{4\eps} - 2\eta/\epsilon$, then $|L| < \frac{1-2\eps}{2\eps} \alpha N$.
\end{claim}

\begin{proof}
Suppose $|L| \geq \frac{1-2\eps}{2\eps} \alpha N$.
Consider $L'\subseteq L$ with $|L'| =  \frac{1-2\eps}{2\eps} \alpha N$.
Let $\delta = \frac{1-2\eps}{2\eps} - x$. 
Notice that $\delta \geq 0$ because    $x \leq  \frac{1-2\eps}{4\eps} - O_\eps(\eta), \eps \leq 1/8$.
Also $x+ \delta \geq 1$ since $\eps \leq 1/8$.

By Lemma \ref{lem:expansion_larger}, $|\Gamma(L')| \geq (1-(x+\delta)\eps)D |L'| - O_\eps(1)$.
By Lemma \ref{lemma:FAddtoLFirst} we can consider $L'$ as being constituted by first adding all elements in $F$ and then add another $\delta \alpha N$ elements by the algorithm.
Notice that all the unsatisfied checks are in $\Gamma(F)$, $|\Gamma(F)| \leq D|F|$, and every element in $L'\setminus F$ contributes at most $2\Delta D$ vertices to $R$. 
Hence $|\Gamma(L')| \leq D|F| + 2\Delta D \delta \alpha N$.
So we have
\[
(1-(x+\delta)\eps)D |L'| - O_\eps(1) \leq  |\Gamma(L')| \leq D|F| + 2\Delta D \delta \alpha N
\]
Thus
\[
(1-(x+\delta)\eps) \cdot (x+\delta) - O_\eps(1/N) \le x + 2 \Delta \delta = x + 2 (\eps + 2\eta) \delta.
\]
So this is equivalent to
\[
(1- 2\eps - \eps (x+\delta))(x+\delta)  - 4 \delta \eta  \le (1-2\eps) x
\]
Recall that $\delta+x=\frac{1 -2 \eps}{2\eps}$. 
To get a contradiction, we only need 
\[
(1-2\eps) x < (1-2\eps)^2/4\eps - 4 \delta \eta.
\]
This is satisfied by $x \leq \frac{1-2\eps}{4\eps} - 2\eta/\epsilon$.

\end{proof}

\begin{proofof}{Theorem \ref{thm:guessexpansion_lineartime}}
One of our enumerations has  $\tilde{\gamma}\tilde{x}$ such that $   \gamma x \in [\tilde{\gamma}\tilde{x} ,\tilde{\gamma}\tilde{x}+ \eta]$.
Consider \textsf{Decoding} in  Algorithm \ref{alg:decoding_small_eps_lineartime} under this enumeration.
After the function $\mathrm{Find}$, all the errors are in $L$ by Claim \ref{claim:FsubseteqL_lineartime}. 

Now we bound $|L|$. 
We can pick the smaller bound of $x$ from Claim   \ref{claim:guessExpansionLbound1_lineartime} and Claim   \ref{claim:guessExpansionLbound2_lineartime}.
If $\eps < \frac{3-2\sqrt{2}}{2} $, then $\frac{\sqrt{2}-1}{2\eps} < \frac{1-2\eps}{4\eps}$. 
So by Claim   \ref{claim:guessExpansionLbound1_lineartime} and Claim   \ref{claim:guessExpansionLbound2_lineartime}  when $x \leq \frac{\sqrt{2}-1}{2\eps} - O(\eta/\eps) $ we have $|L| < \frac{1-2\eps}{2\eps}\alpha N  $.
If $\eps \in [\frac{3-2\sqrt{2}}{2}, 1/8]$, then $\frac{\sqrt{2}-1}{2\eps} \geq \frac{1-2\eps}{4\eps}$.
So by Claim   \ref{claim:guessExpansionLbound1_lineartime} and Claim   \ref{claim:guessExpansionLbound2_lineartime} , when 
$x \leq \frac{1-2\eps}{4\eps} - O(\eta/\eps) $, we have $|L| < \frac{1-2\eps}{2\eps} \alpha N$. 
Since the expander is an $(\alpha N, (1-\eps)D)$ expander,  by Theorem \ref{thm:decfromerasures}, one can correct all the errors efficiently using $L$ (as the set of erasures) and the corrupted codeword.

The decoding algorithm runs in linear time because we only have a constant number of enumerations, and each enumeration takes linear time.
\end{proofof}

%% file: list_decoding_bound.tex

\section{List-decoding Radius}\label{sec:list_decode}
In this section, we consider expander graphs with bounded maximum degree $D_{\max}=O(1)$. Our main result of this section is the following theorem about the list-decoding radius of almost-regular expander codes. For convenience, we only consider relative distance and relative radii. Throughout this section, $\delta=\alpha/2\eps$ denotes the relative distance, $r$ denotes the relative decoding radius from the Johnson bound, and $\rho$ denotes the relative decoding radius that we will prove.
\begin{theorem}\label{thm:list_decoding_radius}
Given any $(\alpha N,(1-\eps)D)$-expander $G$ with a regular degree $D$ in $V_L$ and a maximum degree $D_{\max}$ in $V_R$,  its expander code has a relative list decoding radius $\rho = (\frac{1}{2} + \Omega(1/D_{\max}) ) \cdot \delta$ and list size $N^{O(1)}$.

In particular, when $\eps\le 1/4$, $\alpha/\eps \le 0.1$, and $D_{\max} \le 1.1 D_R$ for the average right degree $D_R$, the relative list-decoding radius $\rho$ is strictly  larger than the Johnson bound $r$ of binary codes with relative distance $\delta=\frac{\alpha}{2\eps}$.
\end{theorem}
We remark that $D_{\max} \le 1.1 D_R$ is a relaxation for $D_R$-regular graphs, which are a standard instantiation of LDPC codes.

We finish the proof of Theorem~\ref{thm:list_decoding_radius} in the rest of this section. First of all, recall that the Johnson bound $r$ of binary codes with relative distance $\delta$ is $\frac{1-\sqrt{1-2\delta}}{2}$, which is the limit of the inequality
\[
\delta/2 + r^2 - r > 0.
\]

To prove Theorem~\ref{thm:list_decoding_radius}, the basic idea is to use locality (which we will define more precisely in the proof) of expander codes to improve the average case in the argument of the Johnson bound. In particular, for $L$ codewords $C_1,\ldots,C_L$ within distance $\rho N$ to some string $y$, we will show that the 1s in $C_1 \oplus y, \ldots, C_L \oplus y$ are concentrated on a constant fraction of positions. More precisely, we pick a threshold $\theta>r$ to show the concentration of 1s. We use the following fact about $\theta$ and $r$ in the proof.



\begin{claim}\label{clm:threshold_theta}
When $\eps \le 1/4$ and $\alpha/\eps \le 0.1$, the relative list decoding radius $r$ of the Johnson bound of relative distance $\delta:=\alpha/2\eps$ of binary codes is less than $0.53 \delta$. Furthermore, when $D_{\max} \le 1.1 D_R$, our choice $\theta:=0.9/D_{\max}$ is at least $0.544 \delta$, which is greater than $r$.
\end{claim}

We defer the proof of Claim~\ref{clm:threshold_theta} to Section~\ref{sec:proof_threshold} and finish the proof of Theorem~\ref{thm:list_decoding_radius} here.

\begin{proofof}{Theorem~\ref{thm:list_decoding_radius}}
We first show $\rho>r$ given $\eps \le 1/4$, $\alpha/\eps \le 0.1$ and $D_{\max} \le 1.1 D_R$ then discuss how to prove $\rho=\delta/2 \cdot (1 + \Omega(1/D_{\max}))$ in general. We fix the threshold $\theta:=0.9/D_{\max}$ as in Claim~\ref{clm:threshold_theta}. 
For convenience, we assume that the decoding radius $\rho$ is always less than $0.54 \delta$ in this proof (otherwise it already satisfies $\rho = \frac{\delta}{2} (1 + \Omega(1/D_{\max})$ and is strictly larger than $r$ from the above claim $r<0.53 \delta$). 
Moreover, $\theta:=0.9/D_{\max} \ge 0.544 \delta$ from Claim~\ref{clm:threshold_theta} is larger than $\rho$.

We fix an arbitrary string $y \in \{0,1\}^N$ and consider the number of codewords within relative distance $\rho$ to it, say, there are $L$ codewords $C_1,\ldots,$ and $C_L$. Let $\Gamma_{odd}(S)$ denote the neighbors of $S$ with an odd number of edges to $S$. Given a $\{0,1\}$-string $z$, let $S_z$ denote the set of 1-entries and $\Gamma_{odd}(z):=\Gamma_{odd}(S_z)$. Back to the $L$ codewords, since $(y\oplus  C_i)\oplus 
(y \oplus C_j)$ is a codeword,  $\Gamma_{odd}(y \oplus C_1)=\cdots=\Gamma_{odd}(y \oplus C_L)$ from the definition of the expander code --- all codewords satisfy those parity checks. Hence we use $\Gamma_{odd}$ to denote this neighbor set $\Gamma_{odd}(y \oplus C_1)=\cdots=\Gamma_{odd}(y \oplus C_L)$. 

First of all, we lower bound $|\Gamma_{odd}|$. We pick $C_i$ such that $|y \oplus C_i| \in [0.5\delta \cdot N, \rho \cdot N]$. Note that such a $C_i$ exists as long as $L \ge 2$. Then $|\Gamma_{odd}(y\oplus C_i)| \ge (1-2\eps \cdot \frac{|y+C_i|}{\alpha N}) D \cdot |y+C_i| - O_{\eps}(1)$ from Lemma~\ref{lem:expansion_larger}, 
which is at least $0.46 \rho D \cdot N - O_{\eps}(1)$ given $\rho \le 0.54 \delta$ and the range of $|y \oplus C_i|$. For ease of exposition, we use the lower bound $|\Gamma_{odd}| \ge 0.45 \rho \cdot D N$ in the rest of this proof.

Let $\tau_i$  denote how many codewords of $C_i$ whose $i$th bit is different from the corresponding bit in $y$, i.e., $\sum_{j=1}^L 1\{i \in \supp(y\oplus C_j)\}$. 
Since $|y \oplus C_i| \le \rho N$, we have $\sum_i \tau_i \le \rho N \cdot L$ --- in another word, $\E_i [\tau_i] \le \rho L$. The key difference between our calculation and the Johnson bound is that we will prove $\tau_1,\ldots,\tau_n$ have a large deviation. 
For convenience, we call $i \in V_L$ heavy if and only if $\tau_i \ge \theta \cdot L$ for $\theta=0.9/D_{\max}$ and show that their sum is $\Theta(N L)$:
\begin{equation}\label{eq:heavy_elements}
S_h:= \sum_{\text{heavy } i} \tau_i \ge 0.45 \rho N \cdot (L -  D_{\max} \cdot \theta L).
\end{equation}
Since $\theta>\rho$ and $\E_i [\tau_i] \le \rho L$, this implies that $\tau_1,\ldots,\tau_N$ have a large deviation. 

To prove Eq~\eqref{eq:heavy_elements}, the starting observation is that for each $v \in \Gamma_{odd} \subseteq V_R$, $\sum_{i \in \Gamma(v)} \tau_i \ge L$, by the definition of $\Gamma_{odd}$. Since $v$ has $\le D_{\max}$ neighbors,
\[
\sum_{\text{heavy } i\in \Gamma(v)} \tau_i \ge L -  D_{\max} \cdot \theta L.
\]
By the double counting argument, \[
\sum_{v \in \Gamma_{odd}} \sum_{\text{heavy }i \in \Gamma(v)} \tau_i \ge (L -  D_{\max} \cdot \theta L) \cdot |\Gamma_{odd}| \ge (L-D_{\max} \cdot \theta L) \cdot 0.45 \rho D \cdot N. 
\]
So 
\[
\sum_{\text{heavy } i} \tau_i \ge \frac{\sum_{v \in \Gamma_{odd}} \sum_{\text{heavy }i \in N(v)} \tau_i}{D} \ge 0.45 \rho N \cdot (L -  D_{\max} \cdot \theta L).
\]
Moreover, let $N_h$ denote the number of heavy elements. We have $\theta L \cdot N_h \le S_h$, which upper bounds $N_h$ by $S_h/\theta L$.

Similar to the argument of the Johnson bound, let $T$ denote all triples of the form $(i,j_1,j_2)$ where $i \in [N]$, $j_1,j_2 \in [L]$ and $C_{j_1}(i) \neq C_{j_2}(i)$. Since the distance between $C_{j_1}$ and $C
_{j_2}$ is at least $\delta N$ for any $j_1 \neq j_2$, the number of triples is at least ${L \choose 2} \cdot \delta N$. 

On the other hand, $T$ is equal to $\sum_{i \in [n]} \tau_i (L-\tau_i)$. Then we provide a upper bound on $\sum_{i \in [n]} \tau_i (L-\tau_i)$ under the two  constraints $\sum_i \tau_i \le \rho N \cdot L$ and $\sum_{\text{heavy } i} \tau_i \ge 0.45 \rho N \cdot (L- D_{\max} \cdot \theta L)$. 

\begin{claim}\label{clm:upper_bound_triples}
Given $\sum_i \tau_i \le \rho N \cdot L$, the heavy threshold $\theta$, and $S_h \ge 0.45 \rho N \cdot (L- D_{\max} \cdot \theta L)$, $\sum_{i \in [n]} \tau_i (L-\tau_i) \le N_h \cdot \theta L (L-\theta L) + (N-N_h) \cdot e L (L - eL) $
where $S_h = 0.45 \rho N \cdot (L- D_{\max} \cdot \theta L)$, $e=\frac{\rho L N - S_h}{L(N-N_h)}$, and $N_h$ is equal to the upper bound $S_h/\theta L$.
\end{claim}

We defer the proof of Claim~\ref{clm:upper_bound_triples} to Section~\ref{sec:bound_triples} and summarize the two bounds to get \[
{L \choose 2} \delta N \le T \le N_h \cdot \theta L (L-\theta L) + (N-N_h) \cdot e L (L - eL) \]
where $e=\frac{\rho L N - S_h}{L(N-N_h)}$ and $N_h=S_h/\theta L$.
This implies 
\[
\left( \delta/2 + \frac{N_h}{N} \theta^2 + \frac{N-N_h}{N} e^2 - \rho \right) L \le \delta/2.
\]
So $L$ is upper bounded when the decoding radius $\rho$ satisfies $\delta/2 + \frac{N_h}{N} \theta^2 + \frac{N-N_h}{N} e^2 - \rho>N^{-O(1)}$ where $\frac{N_h}{N} \theta + \frac{N-N_h}{N} e = \rho$. For convenience, let $\rho^*$ be the limit of $\rho$ satisfying the above inequality, i.e., $\delta/2 + \frac{N_h}{N} \theta^2 + \frac{N-N_h}{N} e^2 - (\rho^*)=0$. We show $\rho^*>r$ and provide an explicit lower bound in the rest of this proof.

\paragraph{Comparing with the Johnson bound.} Recall that the Johnson bound $r$ is obtained from the equation
\[
\delta/2 + r^2 - r = 0,
\]
which implies $r = \frac{1 - \sqrt{1-2 \delta}}{2}$. Back to the equation of $\rho^*$,
\begin{equation}\label{eq:new_radius}
\delta/2 + \frac{N_h}{N} \theta^2 + \frac{N-N_h}{N} e^2 - (\rho^*)=0.
\end{equation}
Given $\frac{N_h}{N} \theta + \frac{N-N_h}{N} e = \rho^*$, the two middle terms
\[
\frac{N_h}{N} \theta^2 + \frac{N-N_h}{N} e^2 = (\rho^*)^2 + \frac{N_h}{N} (\theta - \rho^*)^2 + \frac{N-N_h}{N} (e-\rho^*)^2 = (\rho^*)^2 + \frac{N_h}{N} \cdot \frac{N-N_h}{N} (\theta - e)^2.
\]
This would always increase the range of $\rho^*$ since $\frac{N_h}{N} \cdot \frac{N-N_h}{N} (\theta - e)^2$ is positive. Specifically, the two equations imply
\begin{align*}
(\rho^*)^2 - r^2 + \frac{N_h}{N} \cdot \frac{N-N_h}{N} (\theta - e)^2 - \rho^* + r & = 0 \\
\Leftrightarrow
 (\rho^* - r) \cdot (1 - \rho^* - r) & = \frac{N_h}{N} \cdot \frac{N-N_h}{N} (\theta - e)^2  \\
\Leftrightarrow
\rho^* - r & = \frac{\frac{N_h}{N} \cdot \frac{N-N_h}{N} (\theta - e)^2}{1-\rho^*-r}.
\end{align*}
Since $\theta>0.544 \delta$ and $e<\rho \in [0.5 \delta, 0.54 \delta]$, we have $\theta-e=\Omega( \delta)$. Moreover, $N_h/N=\Omega(D_{\max} \cdot \delta)$ from Claim~\ref{clm:upper_bound_triples} and both $r$ and $\rho$ are at most $0.1$; so we have $\rho^*-r = \Omega(D_{\max} \cdot \delta^3)$.


\paragraph{Showing $\rho^*=\bigg( \frac{1}{2} + \Omega(1/D_{\max}) \bigg) \delta$.} We only need to consider $1/D_{\max} \ge 2\delta$ such that $\theta=0.9/D_{\max}$ is greater than $\rho$ \big(otherwise $r=( \frac{1}{2} + \Omega(1/D_{\max}) ) \delta$ \big). By the same argument, we will get \eqref{eq:new_radius} as the limit of $\rho$. Then we simplify \eqref{eq:new_radius} to
\[
\rho^*>\delta/2 + \frac{N_h}{N} \theta^2 = \delta/2+ \frac{S_h}{N L} \cdot \theta = \delta/2 + 0.045 \rho^* \cdot 0.9 /D_{\max}.
\]
This implies $\rho^* > \frac{d/2}{1 - 0.04 /D_{\max}} = d/2 \cdot (1+\Omega(1/D_{\max}))$\footnote{While a better constant in $\Omega(1/D_{\max})$ is $0.1125$ obtained via $\theta=1/2D_{\max}$, we did not intend to optimize the constant in this work.}.
\end{proofof}

\subsection{Proof of Claim~\ref{clm:threshold_theta}}\label{sec:proof_threshold}
When $\alpha/\eps \le 0.1$, the Johnson bound $r=\frac{1}{2}(1-\sqrt{1-2\delta})$ has a Taylor expansion $\frac{\delta}{2} + \frac{2^{-2}}{2 \cdot 2!} \cdot (2\delta)^2 + \cdots$. This is at most $1.06 \cdot \frac{\delta}{2}=\frac{0.265 \alpha}{\eps}$ for $\delta=\alpha/2\eps$. 

Then, we show $\frac{1}{D_{R}} \ge \frac{0.33 \alpha }{\eps}-O(1/M)$. We plan to apply the 2nd lower bound in Lemma~\ref{lem:expansion_larger} for $k:=0.95/\eps$. 
A subset of size $k \alpha N$ exists because $0.95 \alpha /\eps \le 3.8/D_R + O(1)/M$ from Fact~\ref{fact:relations_parameters}. Since $D_R\ge 4$, this is less than 1 such that one could find a subset of size $k \alpha N$ in $V_L$. Next we apply Lemma~\ref{lem:expansion_larger} to such a  subset in $V_L$ of size $k \alpha N$ and have
\[
\frac{k}{2} (1- \frac{2k\eps -1}{3-2/k}) \cdot D \alpha N - O(1) \le M.
\]
For $k=0.95/\eps$, we use $D N = D_R M$ to simplify it to
\[
\frac{0.95}{2 \eps} \cdot (1- \frac{0.9}{3-2\eps/0.9}) \cdot \alpha D_R M - O(1) \le M.
\]
Since $\eps \le 1/4$, we have 
\[
\frac{0.95 \alpha}{2 \eps} \cdot (1-\frac{0.9}{3}) \le 1/D_R + O(1/M),
\]which shows $1/D_R \ge \frac{0.3325 \alpha}{\eps} - O(1/M)$

Given $D_{\max} \le 1.1 D_R$, we have that $\theta:=0.9/D_{\max} \ge 0.9/1.1 D_R \ge 0.272 \alpha /\eps$ is strictly larger than $r<0.265\alpha/\eps$.

\subsection{Proof of Claim~\ref{clm:upper_bound_triples}}\label{sec:bound_triples}
We divide the proof into four steps:
\begin{itemize}
    \item When $\sum_{i} \tau_i$, $S_h$ and $N_h$ are fixed, $\sum_{i \in [n]} \tau_i (L-\tau_i)$ is maximized at $\tau_i=S_h/N_h$ for heavy elements and $\tau_i=(\sum_{i} \tau_i- S_h)/(N-N_h)$ for non-heavy elements. So we assume heavy elements and non-heavy elements have the same values of $\tau_i$ separately.

    \item Then we fix $S_h$ and $N_h$ and focus on $\sum_i \tau_i$. Since $\tau_i=(\sum_{i} \tau_i- S_h)/(N-N_h)$ for non-heavy $i$'s is less than $L/2$, increasing $\sum_i \tau_i$ will make $\sum_i \tau_i (L-\tau_i)$ larger. So we assume $\sum_i \tau_i=\rho L N$ to estimate an upper bound.
    
    \item Next, when $S_h$ is fixed, the upper bound 
    \begin{equation}\label{eq:upper_bound_johnson}
        N_h \cdot (S_h/N_h) \cdot (L-S_h/N_h) + (N-N_h) \cdot \frac{\rho L N - S_H}{N-N_h} \cdot (L-\frac{\rho L N - S_H}{N-N_h})
    \end{equation} has a positive derivative $(\frac{S_h}{N_h})^2 - (\frac{\rho L N - S_h}{N-N_h})^2$ with $N_h$ from the definition of heavy elements. To estimate an upper bound, we fix $N_h=S_h/\theta L$.
    
    \item Finally, by the convexity of $\sum_i \tau_i^2$, the upper bound in \eqref{eq:upper_bound_johnson} is maximized at $S_h=0.45 \rho N (L - D_{\max} \cdot \theta L)$.
\end{itemize}
So we obtain a upper bound where $N_h=S_h/\theta L$ heavy elements have $\tau_i=\theta L$ and the rest elements have $\tau_i=\frac{\rho L N-S_h}{N-N_h}$.

%% file: open.tex
\section{Open Questions}\label{sec:open}
Our work leaves many intriguing open questions, and we list some of them here.
\begin{enumerate}
    \item Our distance in Theorem~\ref{thm:infor_distance} is only shown to be tight by a graph that is not strictly regular on the right. For bipartite expander graphs that are regular on both sides, is it possible to get an improved distance bound, or is the bound in Theorem~\ref{thm:infor_distance} still tight?
    \item Can one design efficient algorithms to correct more errors? Specifically, it would be nice to get close to the half distance bound. Alternatively, is there any hardness result that prevents us from achieving this?
    \item Can one design efficient algorithms to correct more errors for the case of $\eps \ge 1/4$? So far all our improvements over previous results are only for the case of $\eps < 1/4$.
    \item Can one get a better list decoding radius for general expander codes? Can one design efficient list decoding algorithms? As mentioned before, any efficient list decoding algorithm would also immediately improve our results on unique decoding, and in particular imply unique decoding up to half distance. If there is any hardness result for unique decoding close to half distance, this would also rule out the possibility of list decoding for general expander codes.
\end{enumerate}

%% file: appendix.tex
\section{Supplemental Proofs}\label{appen:proof_expansion}

We finish the calculation omitted in Section~\ref{sec:bound_dist} here. We provide one calculation for graphs that is not necessarily regular on the right and another calculation for regular graphs.
\begin{proposition}\label{appendix:random_graph}
    If parameters $\alpha, \epsilon, M, N, D$ satisfies $\left( \frac{e}{\alpha} \right)\cdot\left(\frac{e\alpha ND}{\epsilon M}\right)^{\epsilon D} < 1$, then the probability of a random bipartite graph, where each vertex in $V_L$ has $D$ random neighbros, is $(\alpha N, (1-\epsilon)D)$--expander is close to 1. 
\end{proposition}

\begin{proof}
    Suppose the left part of the bipartite graph is $[N]$. Fix a subset $X$ of $[N]$ with size $\alpha N$, and let $y_i^1, \cdots, y_i^D$ be the neighbours of the $i$--th vertex in $X$. Then the expansion of $X$ is less than $(1-\epsilon)D$ is equivalent to $\#\left\{y_i^j\right\} < (1-\epsilon)D\alpha N$, where $i \in X$ and $j \in [D]$. 
    
    Arrange $y_i^j$ in the lexicographic order of $(i, j)$. The probability of the value of $y_i^j$ has been taken before it does not exceeds $\frac{\#\left\{y_{i'}^{j'}\ \middle|\ (i', j') \prec (i, j)\right\}}{M} < \frac{\alpha ND}{M}$. 
    
    So the probability of the expansion of $X$ is less than $(1-\epsilon)D$ is less than $\binom{\alpha ND}{\epsilon \alpha ND}\cdot \left(\frac{\alpha ND}{M}\right)^{\epsilon \alpha ND}$. 
    
    Hence, the probability of the random graph is not $(\alpha N, (1-\epsilon)D)$--expander is less than
    \begin{equation}\label{appendix:formula}
        \binom{N}{\alpha N}\cdot \binom{\alpha ND}{\epsilon \alpha ND}\cdot \left(\frac{\alpha ND}{M}\right)^{\epsilon \alpha ND}
    \end{equation}
    
    By the approximation of binomial coefficient: $\binom{A}{B} < \left( \frac{eA}{B} \right)^B$, $(\ref{appendix:formula})$ is less than
    \[ \left(\frac{eN}{\alpha N}\right)^{\alpha N} \cdot \left( \frac{e\alpha ND}{\epsilon\alpha ND} \right)^{\epsilon \alpha ND} \cdot \left(\frac{\alpha ND}{M}\right)^{\epsilon \alpha ND} = \left(\left( \frac{e}{\alpha} \right)\cdot \left( \frac{e\alpha ND}{\epsilon M} \right)^{\epsilon D}\right)^{\alpha N} \]
\end{proof}

Given any constant $\eps \in (0, 1)$, by choosing a large enough constant $D$ and let $D_R=\frac{D N}{M}$ be the average degree on the right, Proposition~\ref{appendix:random_graph} immediately implies the following proposition.

\begin{proposition}\label{prop:parameterrelation}
For any constants $\eps, \eta \in (0, 1)$, there exist a constant $D$ and $(\alpha N, (1-\eps)D)$-expanders such that $\frac{\alpha}{\eps} \geq \frac{1/e-\eta}{D_R}$.
\end{proposition}

One can also obtain a regular expander by choosing an integer $D_R=\frac{D N}{M}$ and generating $D_R$ permutations. Such a random expander has been proved in \cite{SS96}. We provide an argument for completeness. 

Here is a technical lemma summarized from \cite{SS96}.
\begin{proposition}\label{Prop:regular}
	Let $B$ be a random $(D, D_R)$--regular bipartite graph with left size $N$ and right size $\frac{D \cdot N}{D_R}$. Then for all $0 < \alpha < 1$, with exponentially high probability all sets of $\alpha n$ vertices in the left part have at least
	\[ N\left( \frac{D}{D_R}\left( 1-(1-\alpha)^{D_R} \right) - 2\alpha \cdot \sqrt{D \log e/\alpha}\right) \]
	neighbours.
\end{proposition}

Before we prove this proposition, we show how to choose the parameters to make the expansion at least $(1-\eps)D$. Recall that in the proof of Theorem~\ref{thm:tight_distance} in Section~\ref{sec:bound_dist}, we are looking at a random bipartite graph with $N_1=N-N' \ge N/2$ left vertices, $M_1=M-DN'/2$ right vertices, regular left degree $D$ and regular right degree $D_R=N_1 \cdot D/M_1$. Since $M_1 \ge M/2 \ge N/4$ and $N_1 \le N$, we have $D_R \le 4D$. Next we choose $\alpha = 10^{-3} \cdot (\eps/D)^2$ such that for any $\alpha' \leq 2 \alpha$, $(1-\alpha')^{D_R} \in [1-\alpha' D_R, 1-(1-\eps/2)\alpha' D_R]$ and  $1-(1-\alpha')^{D_R} \in \big[  (1-\eps/2) \alpha' D_R, \alpha' D_R \big]$. Note that any subset of size $\alpha N$ has size $\alpha' N_1$ with $\alpha \le \alpha' \le 2 \alpha$. Thus we simplify the bound in the above proposition to get the desired expansion
\begin{align*}
    & N_1\left( \frac{D}{D_R} \cdot (1-\eps/2) \alpha' D_R - 2\alpha' \cdot \sqrt{D \log (e/\alpha')}\right) = N_1 D \alpha' \cdot \left(1-\eps/2 - 2\sqrt{\frac{\log (e/\alpha')}{D}} \right) \\ \ge & N_1D \alpha' \cdot (1-\eps) = \alpha ND \cdot (1-\eps),
\end{align*}
for a sufficiently large constant $D=D(\eps)$.


\begin{proof}[Proof of Proposition~\ref{Prop:regular}]
	First, we fix a set of $\alpha N$ vertices in the left part, $V$, and estimate the probability that $\Gamma(V)$ is small. The probability of a certain vertex in the right part is contained in $\Gamma(V)$ is at least $1 - (1-\alpha)^{D_R}$. Thus the expected number of neighbours of $V$ is at least $M \cdot (1 - (1-\alpha)^{D_R}) = \frac{nD\left(1-(1-\alpha)^{D_R}\right)}{D_R}$. We will use Azuma inequality to derive that $|\Gamma(V)|$ has a small deviation property, and hence the probability of $|\Gamma(V)|$ less than the expectation minus some deviation is exponentially small. 
	
	Actually, we number the edges outgoing from V by 1 through $D \alpha N$. Let $X_i$ be the random variable of the expected size of $|\Gamma(V)|$ given the choice of the first $i$ edges leaving from $V$. Clearly, $X_1, \cdots, X_{D \alpha N}$ form a martingale and $|X_{i+1}-X_i| \leqslant 1$. 
	
	By Azuma's inequality, we have:
	\[ \mathbb{P}\left( \mathbb{E}\left(X_{D \alpha n}\right) - X_{D \alpha N} > \lambda \sqrt{D \alpha N} \right)  < \exp\left( -\lambda^2/2 \right) \]
	
	Since there are $\binom{N}{\alpha N}$ choices for the set $V$, it suffices to choose $\lambda$ such that
	\[ \binom{N}{\alpha N}e^{-\lambda^2/2} \text{ is exponentially small}. \]
    Since ${N \choose \alpha N} \le (e/\alpha)^{\alpha N}$, we choose $\lambda=2 \cdot \sqrt{\alpha N \cdot \log (e/\alpha)}$ to make it exponentially small. Then the deviation becomes 
    \[
    \sqrt{D \alpha N} \cdot 2 \sqrt{\alpha N \cdot \log (e/\alpha)}=2\alpha N \cdot \sqrt{D \log (e/\alpha)}
    \]
\end{proof}